\title{Flexible control of the median of the false discovery proportion} 

\author{Jesse Hemerik\footnote{Econometric Institute, Erasmus University, Burg. Oudlaan 50,
3062 PA Rotterdam, The Netherlands. e-mail: hemerik@ese.eur.nl}
, Aldo Solari\footnote{Department of Economics, Ca' Foscari University, Cannaregio 873, 30121 Venice, Italy.} \phantom{.}and
Jelle J. Goeman\footnote{Department of Biomedical Data Sciences, Leiden University Medical Center, Einthovenweg 20,
2333 ZC Leiden, The Netherlands.}
}

\documentclass[11pt]{article}

\usepackage{amsmath}
\usepackage{amsfonts}
\usepackage{bm}
\usepackage{bbm}
\usepackage{amsthm}
\usepackage{comment}
\usepackage[nottoc]{tocbibind}
\usepackage{graphicx}

\usepackage{subcaption}

\usepackage{changepage}
\usepackage{natbib}
\usepackage{abstract}
\usepackage[colorlinks=true,citecolor=blue,runcolor=blue,linkcolor=blue, pdfborder={0 0 0}]{hyperref}

\usepackage{tikz}
\usepackage{pgfplots}

\usepackage{algorithm}
\usepackage[noend]{algorithmic}

\theoremstyle{plain}
\newtheorem{theorem}{Theorem}

\newtheorem{proposition}{Proposition}

\newtheorem{example}{Example}

\theoremstyle{definition}

\newtheorem{assumption}{Assumption}

\usepackage[margin=3cm]{geometry}

\newcommand{\N}{\mathcal{N}}
\newcommand{\R}{\mathcal{R}}
\newcommand{\X}{\mathcal{X}}
\newcommand{\C}{\mathcal{C}}

\definecolor{darkblue}{rgb}{0.0, 0.2, 0.6}

\begin{document}
\maketitle

\begin{abstract}
\noindent 
We introduce a multiple testing procedure that controls the median of the proportion of false discoveries (FDP) in a flexible way. The procedure only requires a vector of p-values as input and is comparable to the Benjamini-Hochberg method, which controls the mean of the FDP. Our method allows freely choosing one or several values of alpha after seeing the data -- unlike Benjamini-Hochberg, which can be very liberal when alpha is chosen post hoc. We prove these claims and illustrate them with simulations.
Our procedure is inspired by a popular estimator of the total number of true hypotheses. We adapt this estimator to provide  simultaneously median unbiased estimators of the FDP, valid for finite samples. This simultaneity allows for the claimed flexibility.  Our approach does not assume independence. 
The time complexity of our method is linear in the number of hypotheses, after sorting the p-values.
\\
\\
\emph{keywords:} control;  estimation;  false discovery proportion; false discovery rate; post hoc
\end{abstract}

\section{Introduction}
Multiple hypothesis testing procedures have the common aim of ensuring that the number of incorrect rejections, i.e. false positives, is likely small. 
The most commonly used multiple testing procedures control either the \emph{family-wise error rate} or  the
\emph{false discovery rate} (FDR) \citep{dickhaus2014simultaneous,harvey2020evaluation}. The false discovery rate is the expected value of the \emph{false discovery proportion} (FDP), which is the proportion of false positives among all rejections of null hypotheses. Controlling the FDR means ensuring that the expected FDP is kept below some pre-specified value $\alpha$ \citep{benjamini1995controlling,benjamini2001control,goeman2014multiple}.

The FDP, which is an unknown quantity, can vary widely about its mean, when the tested variables are strongly correlated \citep{efron2007correlation,schwartzman2011effect,delattre2015new}. 
For this reason, methods have been developed that do not control the FDR or estimate the FDP, but rather provide a confidence interval for the FDP \citep{hemerik2018false}. Some methods provide confidence intervals for several choices of the set of rejected hypotheses, that are simultaneously valid with high confidence \citep{genovese2004stochastic,genovese2006exceedance,meinshausen2006false,hemerik2019permutation, katsevich2020simultaneous, blanchard2020post, goeman2021only,blain2022notip, vesely2023permutation}. 
There are also procedures, including the methods just mentioned, that ensure that the FDP remains small with high confidence \citep{van2004augmentation,lehmann2005generalizations, romano2007control, guo2007generalized, farcomeni2008review, roquain2011type, guo2014further, delattre2015new, ditzhaus2019variability,
dohler2020controlling, basu2021empirical, mie2022}. This is often termed ``false discovery exceedance control.''

Methods that ensure that the FDP remains small with high confidence can provide very clear and useful error guarantees. The downside of these methods, however, is that under dependence, they often do not  have sufficient power to reject any hypotheses, even if there is substantial signal in the data. The reason is that these methods do not merely require that the FDP is small on average, but small with high confidence. As a result, users may prefer approaches with weaker guarantees, such as FDR methods. 
 An alternative is to take $\alpha=0.5$ in FDP methods.

The most popular FDR method is the Benjamini-Hochberg method (BH)  \citep{benjamini1995controlling}.
FDR methods generally require the user to choose $\alpha$ before looking at the data. Common choices for $\alpha$ are $0.05$ and $0.1$. The methods guarantee that the FDR is kept below $\alpha$. However,  researchers would often like to change $\alpha$ post hoc. For example, if no hypotheses are rejected for $\alpha=0.05$, a researcher may want to increase $\alpha$ to 0.1, changing the FDP target in order to obtain more rejections. 
In other cases, the user will  want to decrease $\alpha$.
However, as we show (in Section \ref{secBHnotflexible} and the Supplementary Material), choosing $\alpha$ post hoc can severely invalidate methods such as BH.
Moreover, the user may want to report results for several values of $\alpha$, while providing a simultaneous error guarantee.
There is a need for methods that allow for these types of inference.

In this paper, we provide a class of multiple testing methods that allow to choose the threshold freely after looking at the data. Our methodology only requires a vector of \emph{p}-values as input, is non-asymptotic.
Our procedure controls the \emph{median} of the FDP rather than the \emph{mean}. For this and other reasons, we denote the \emph{target FDP} by $\gamma\in[0,1]$ instead of $\alpha$ \citep[inspired by][]{romano2007control,harvey2020evaluation,basu2021empirical}. Controlling the median means that the FDP is at most $\gamma$ with probability at least $0.5$. We will refer to this as \emph{mFDP control}.  Note that mFDP control can also be obtained with several existing FDP methods, by taking $\alpha=0.5$. 
Like some existing methods,  our procedure is flexible in the sense that $\gamma$ can be freely chosen after seeing the data.
Further, our procedure is adaptive, in the sense that it does not necessarily become conservative if the fraction of false hypotheses is large. 
We prove that our procedure is valid under a novel type of assumption on the joint distribution of the \emph{p}-values. 
  In particular, our method does not require independence. 
Moreover, the method was valid in all simulation settings considered.
Further, we prove that  our procedures are often admissible, i.e., they cannot be uniformly improved \citep{goeman2021only}. 
Since the method in \citet{goeman2019simultaneous} is also flexible and admissible in some settings, we compare with that method in simulations. We also compare  with the elegant and fast method from \citet{katsevich2020simultaneous}. 
Our methodology has been implemented in the \verb|R| package \verb|mFDP|, available on CRAN.

Our procedure is partly inspired by an existing estimator of the fraction $\pi_0\in[0,1]$ of true hypotheses among all hypotheses. This estimator is mentioned in  \citet{schweder1982plots} and advocated in \citet{storey2002direct}.  We will refer to it as the Schweder-Spj{\o}tvoll-Storey estimator. Some  publications  refer to it as Storey's estimator or the Schweder-Spj{\o}tvoll estimator \citep{hoang2022usage}. 
The literature proposes multiple $\pi_0$ estimators based on \emph{p}-values \citep{rogan1978estimating,hochberg1990more,langaas2005estimating,meinshausen2006estimating,rosenblatt2021prevalence}.
As a side result of our investigation of $\pi_0$ and FDP estimation, we add to this literature a novel $\pi_0$ estimator that is slightly different from Schweder-Spj{\o}tvoll-Storey, unless its tuning parameter is  $0.5$.

The proposed methodology also draws from an idea in \citet{hemerik2019permutation}, which is to construct simultaneous FDP bounds, called \emph{confidence envelopes}, in a manner that is partly data-based and partly reliant on a pre-specified family of candidate envelopes. The simultaneity of the constructed bounds allows for post hoc selection of rejection thresholds and hence post hoc specification of $\gamma$. The methodology proposed here is applicable in many situations where the method in \citet{hemerik2019permutation} is not. The reason is that one cannot generally use permutations if one only has \emph{p}-values, which is the setting we assume.

Our mFDP controlling approach  conceptually relates to recent methods that bound the FDR by $\alpha$ by finding the largest \emph{p}-value threshold for which some conservative estimate of the FDP is below $\alpha$ \citep{barber2015controlling,li2017accumulation,lei2018adapt,luo2020competition,lei2021general,rajchert2022controlling}. Those methods do not offer the simultaneity provided in the present paper.

In Section \ref{secestimation}, we start with non-simultaneous estimation of  the number of false positives.  Section \ref{seccontrol} contains the main theoretical results, which build on Section \ref{secestimation}.
In Section \ref{secsims} we use simulations to  investigate properties of our method. 
We find that the method was valid in all considered simulation settings and had   good power compared to competitors, especially in  high-dimensional settings  with many false hypotheses.
The Supplementary Material contains proofs, additional simulations, an analysis of real RNA-Seq data and theoretical extensions of our results.

\section{Median unbiased estimation of the FDP} \label{secestimation}

\subsection{Notation}
Throughout this paper we consider hypotheses $H_1,...,H_m$ and corresponding \emph{p}-values $p_1,...,p_m$, which take values in $(0,1]$. Write $p=(p_1,...,p_m)$. Let $\N=\{1\leq i \leq m: \text{} H_i \text{ is true}\}$ be the set of indices of true hypotheses and let $N=|\N|$ be the number of true hypotheses, which we assume to be strictly positive for convenience. The fraction of true hypotheses is $\pi_0=N/m$. Let $q_1,...,q_N$ denote the the \emph{p}-values corresponding to the true hypotheses, in any order.  Write $q=(q_1,...,q_N)$.

If $t\in(0,1)$, we write $\R(t)=\{1\leq i \leq m: p_i\leq t\}.$ 
We will call $\R=\R(t)$ the set of rejected hypotheses, since $t$ will usually denote the \emph{p}-value threshold. Write $R=|\R|$.
Let $V=|\N\cap\R|$ be the number of true hypotheses in $\R$, i.e., the number of false positive findings.
We write  $a\wedge b$ for the minimum of numbers $a$ and $b$.

\subsection{The Schweder-Spj{\o}tvoll-Storey  estimate} \label{secintrostorey}
The paper's first results, which inspired Section \ref{seccontrol}, follow from a reinvestigation of the Schweder-Spj{\o}tvoll-Storey estimator of  $\pi_0$  \citep{schweder1982plots,storey2002direct}. 
The estimator depends on a tuning parameter in $(0,1)$ that is usually denoted by $\lambda$. For practical reasons we will write the estimator  in terms of $t:=1-\lambda$.  
The estimator is
\begin{equation} \label{eqstorey}
\hat{\pi}_0':=\frac{|\{1\leq i \leq m:p_i>\lambda\}|}{m(1-\lambda)} = \frac{|\{1\leq i \leq m:p_i>1-t\}|}{mt}
\end{equation}

The heuristics behind this estimate are as follows. The non-null \emph{p}-values, i.e., the \emph{p}-values corresponding to false hypotheses, tend to be smaller than $1-t$, so that most of the \emph{p}-values larger than $1-t$ are null \emph{p}-values. Since for point null hypotheses the null p-values are standard uniform, one expects about $t\cdot100\%$ of the null \emph{p}-values to be larger than $1-t$. Hence, a (conservative) estimate of the number of null \emph{p}-values is $t^{-1}|\{i:p_i>1-t\}|$. Thus, $\hat{\pi}_0'$ is an estimate of $\pi_0$. Storey's estimator is related to the concept of accumulation functions, used to estimate false discovery proportions \citep{li2017accumulation, lei2021general}.

 Note that $\hat{\pi}_0'$ can be larger than 1. Consequently, researchers often use 
 $ \hat{\pi}_0:=\hat{\pi}_0'\wedge 1$. This estimate is usually no longer biased upwards, but downwards for large values of $\pi_0$, in particular $\pi_0=1$.

\subsection{Median unbiased estimation of  $V$  and $\pi_0$} \label{secmedunb}

 Here we derive estimators  of  $V$ and $\pi_0$ that are inspired by the  Schweder-Spj{\o}tvoll-Storey estimator. We make the following assumption. 

\begin{assumption} \label{as1}
The following holds: 
\begin{equation} \label{eqas1}
\mathbb{P}\Big\{ \big|\big\{1\leq i \leq N: q_i\leq t\big\}\big|> \big|\big\{1\leq i \leq m: p_i\geq 1-t\big\}\big|    \Big\}\leq 0.5.
\end{equation}
\end{assumption}
Assumption \ref{as1} says that the number of small ($\leq t$) null p-values tends to be smaller than the number of large ($\geq 1-t$) \emph{p}-values (null and non-null).
Note that this assumption is satisfied in particular if
\begin{equation} \label{eqas2}
\mathbb{P}\Big\{ \big|\big\{1\leq i \leq N: q_i\leq t\big\}\big|> \big|\big\{1\leq i \leq N: q_i\geq 1-t\big\}\big|    \Big\}\leq 0.5.
\end{equation}
Further, note that the probability in \eqref{eqas2} is equal to 
\begin{equation} \label{eq2}
\mathbb{P}\Big\{ \big|\big\{1\leq i \leq N: q_i\leq t\big\}\big|> \big|\big\{1\leq i \leq N: 1-q_i\leq t\big\}\big|    \Big\}.
\end{equation}

If the null \emph{p}-values $q_1,...,q_N$ are independent and standard uniform, then Assumption \ref{as1} is clearly satisfied.
As another example, suppose $q=(q_1,...,q_N)$ is symmetric about $1/2$, i.e., 
\begin{equation} \label{eqsymm}
(q_1,...,q_N) \,{\buildrel d \over =}\, (1-q_1,...,1-q_N).
\end{equation}
Then property \eqref{eq2} and hence Assumption \ref{as1} also hold.
The symmetry property \eqref{eqsymm} holds for instance  if $q_1,...,q_N$ are left- or right-sided \emph{p}-values from  Z-tests based on test statistics $Z_1,...,Z_m$ with joint $\N(0,\Sigma)$ distribution.   
 Further, note that  null \emph{p}-values that are stochastically larger than uniform, or  the presence of many non-null’s makes it easier for Assumption \ref{as1} to be satisfied.

 
Note that if $t$ is used as a rejection threshold, the number of false positive findings is
$$ V(t):=  \big|\big\{1\leq i \leq N: q_i\leq t\big\}\big|.$$
Under  Assumption \ref{as1},
 with probability at least $0.5$, we have
  \begin{equation} \label{overlineV}
  V(t) \leq  \overline{V}(t):= 
  |\{1\leq i \leq m: p_i\geq 1-t\}|.
 \end{equation}
In other words, $\overline{V}(t)$ is a $50\%$-confidence upper bound for $V(t)$. We will refer to such bounds as \emph{median unbiased} estimators for brevity, although writing  `not-downward biased' instead of `unbiased' would be more precise.

 This result also leads to a median unbiased estimator of $\pi_0$. Indeed,
if $V \leq \overline{V} $, then $\R$ contains at least $R-\overline{V}$ false hypotheses, so that $\pi_0$ is at most 
 $$\frac{m- R+\overline{V} }{m}= \frac{m- |\{1\leq i \leq m: p_i\leq t\}|+\big|\big\{1\leq i \leq m: p_i\geq 1-t\big\}\big|  }{m}.$$
A rewrite gives the following result.

 \begin{theorem}
 Suppose Assumption \ref{as1} is satisfied.
 Then $\overline{V}(t)$, defined at \eqref{overlineV}, is a median unbiased estimate of $V(t)$.
As a consequence, $\overline{\pi}_0:= \overline{\pi}_0'\wedge 1$, where
  $$\overline{\pi}_0' =  \frac{|\{1\leq i \leq m: p_i> t\}|+\big|\big\{1\leq i \leq m: p_i\geq 1-t\big\}\big|  }{m},$$ is a median unbiased estimate of $\pi_0$.
 Further,  if $t=0.5$ and no p-value equals $t$, then $\overline{\pi}_0'$ is equal to the Schweder-Spj{\o}tvoll-Storey estimate 
$\hat{\pi}_0'$.
    \end{theorem}

      Thus, if the \emph{p}-values are continuous and $t=0.5$, then $\overline{\pi}_0'=\hat{\pi}_0'$ with probability 1. For other values of $\lambda$, we obtain a median unbiased estimate $\overline{\pi}_0'$ that is slightly different from $\hat{\pi}_0'$. 
 In the Supplementary Material, we provide a theoretical comparison of $\mathbb{E}(\overline{\pi}_0')$ versus $\mathbb{E}(\hat{\pi}_0')$.
In the Supplementary Material we also  obtain the estimate $\overline{\pi}_0'$ in an alternative way and, doing so, discover a broader class of $\pi_0$ estimators.

 We write  $\overline{\pi}_0= \min\{\overline{\pi}_0',1\}.$
 In Example \ref{rex1} and the corresponding Figure \ref{fig:Storey}, the Schweder-Spj{\o}tvoll-Storey method is applied to 500 simulated \emph{p}-values.

      \begin{example}[Running example, part 1: estimating $\pi_0$ and $V$]  \label{rex1}
As a toy example we generated 500 independent \emph{p}-values, 400 of which were uniformly distributed on $[0,1]$ and 100 of which were stochastically smaller than uniform on $[0,1]$. Thus, we can say that $N=400$. A scatterplot of the sorted \emph{p}-values is shown in Figure \ref{fig:Storey}, as well as a visual illustration of how Storey's estimate $\hat{\pi}_0 m$ of the number of true hypothes is computed, in case $\lambda=1-t=0.8$. Often $\lambda$ is taken smaller, but considering small $t$ instead will turn out to be useful. In this example, Storey's estimate $\hat{\pi}_0 \cdot m$ was 410 and our estimate, which is less easy to visualize, was $\overline{\pi}_0 \cdot m=402$. Thus, the estimates were close, as is often the case. 
Since property \eqref{eqsymm} and hence Assumption \ref{as1} is satisfied, we know that $\overline{\pi}_0$ is a median unbiased estimator of  $\pi_0$. In particular, we know with $50\%$ confidence that there are at least $500-402=98$ false hypotheses in total.

As explained in this section, we can make this statement stronger by
noting that $R(t)=180$ and $\overline{V}(t)=82$. The latter means that we know with $50\%$ confidence that there are at least $180-82=98$ false hypotheses among the hypotheses with p-values below $t=0.2$.
\end{example}

\begin{figure}[ht] 
\centering
  \includegraphics[width=0.8\linewidth]{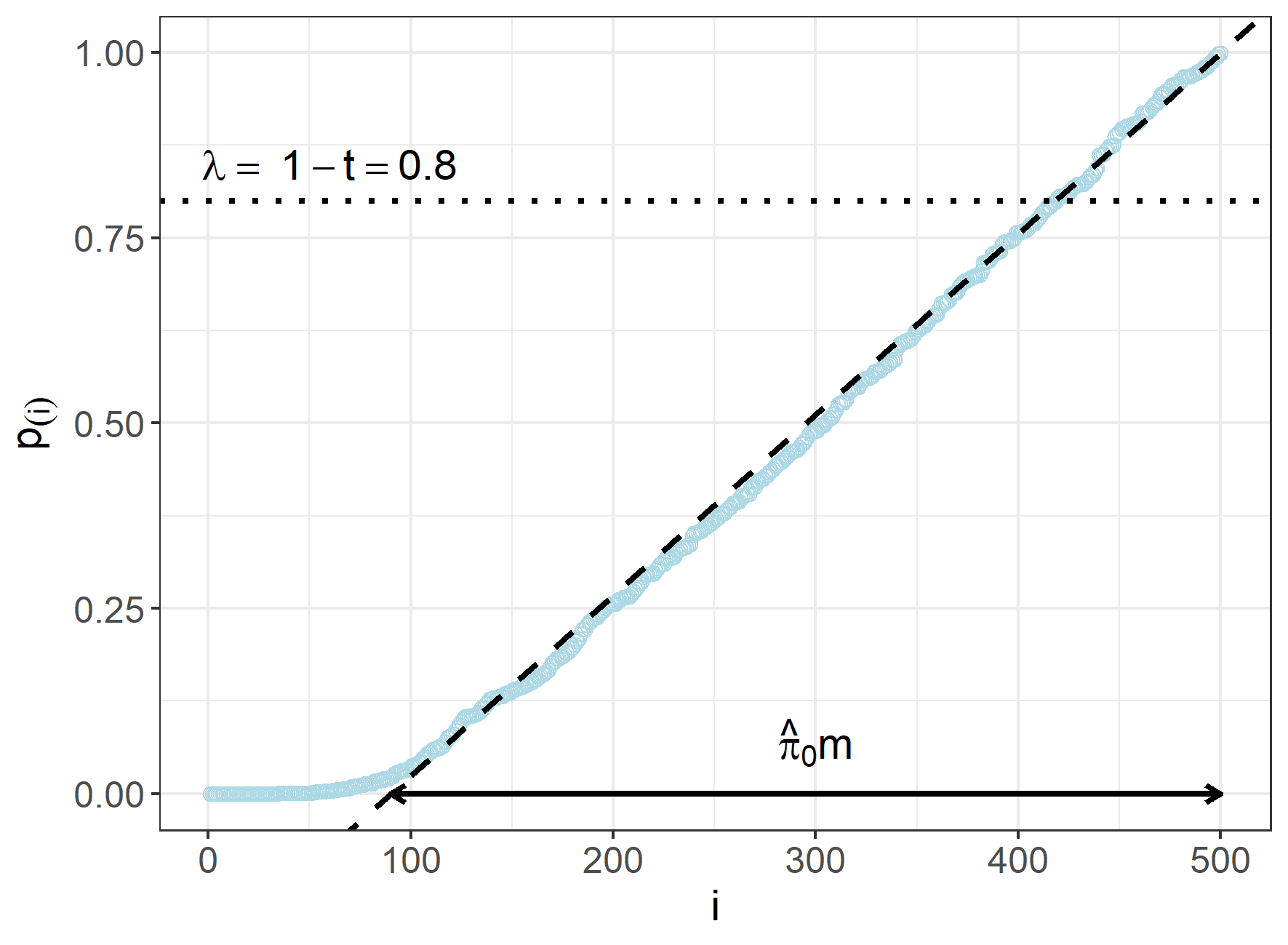}
  \caption{Illustration of the computation of the Schweder-Spj{\o}tvoll-Storey estimate $\hat{\pi}_0$, based on 500 sorted simulated \emph{p}-values. The dashed, straight line is constructed in such a way that it goes through both (500,1) and the point where the dotted line intersects the curve of \emph{p}-values, roughly speaking.} \label{fig:Storey}
\end{figure}

\subsection{Median unbiased estimation of the FDP} \label{FDPfixedc}

  Define the FDP to be the proportion of false positives,
     $$FDP=\frac{V}{R},  \quad  FDP(t)=\frac{V(t)}{R(t)},$$
     which is understood to be $0$ when $R=0$.  
      The median unbiased estimate $\overline{V}$ immediately implies a median unbiased estimate of the FDP. 
    \begin{theorem}
     Suppose Assumption \ref{as1} is satisfied.
   The variable $\overline{FDP}(t)=\overline{V}(t)/R(t)$ is a median unbiased estimator for the FDP, i.e.,
   \begin{equation} \label{medianunb}
   \mathbb{P}\Big\{FDP(t)\leq \overline{FDP}(t)\Big\}\geq 0.5.
   \end{equation}
    \end{theorem}
    To prove this, we only need to remark that if $V\leq \overline{V}$, then  $FDP\leq \overline{FDP}$.

\section{Controlling the mFDP} \label{seccontrol}

\subsection{Overview of our method and comparison with FDR control} \label{secoverv}
In section \ref{FDPfixedc} we considered a fixed rejection threshold $t$ and provided a median unbiased estimate for $FDP(t)$. In many situations, one would like to adapt the threshold $t$ based on the data, in such a way that one still obtains a valid median unbiased estimate. Note that naively choosing $t$ in such a way that an attractive (low) estimate of the FDP is obtained, can invalidate the procedure, in the sense that inequality \eqref{medianunb} not longer holds. In Section \ref{secsimbounds} however, we derive a method that provides median unbiased bounds for a large range of $t$, in such a way that with probability at least $0.5$, the bounds are simultaneously valid for all $t$. 

Specifically, we let the user choose some range $\mathbb{T}\subseteq [0,1]$ of rejection thresholds $t$ of interest, before looking at the data. 
Usually a good choice for $\mathbb{T}$ will be $[0,1/2]$ or another interval starting at $0$.
Then we provide $50\%$-confidence upper bounds $B(t)$ for $V(t)$ that are simultaneously valid over all $t \in \mathbb{T}$:
\begin{equation} \label{envelope}
 \mathbb{P}\Big[\bigcap_{t\in \mathbb{T}} \{V(t)\leq B(t)\}\Big]\geq 0.5.
\end{equation}
It then immediately follows that  $B(t)/R(t)$, $t\in \mathbb{T}$, are simultaneously valid $50\%$-confidence bounds for $FDP(t)$:
$$ \mathbb{P}\Big[\bigcap_{t\in \mathbb{T}} \{FDP(t)\leq B(t)/R(t)\}\Big]\geq 0.5.$$

Since the threshold $t$ can be chosen based on the data, it can be picked such that $B(t)/R(t)$ is low. In particular, one can prespecify a value $\gamma\in[0,1]$, for example $\gamma=0.05$, and take the threshold $t\in \mathbb{T}$ to be the largest value for which  $B(t)/R(t)\leq \gamma$, if such a $t$ exists. This means that our method can be used to reject a set of hypotheses in such a way that the median of the FDP is bounded by $\gamma$:
$$\mathbb{P}(FDP\leq \gamma)\geq 0.5.$$
In other words, we can control the median of the FDP, which we will call the \emph{mFDP}. 
Our method is an example of  false discovery exceedance control, but with the added property that $\gamma$ can be chosen post hoc, as we discuss below.
Our notation `$\gamma$' is in line with e.g. \citet{romano2007control} and \citet{basu2021empirical}.

Our method is related to the popular BH procedure, which ensures that $\mathbb{E}(FDP)\leq \gamma$ \citep{benjamini1995controlling}. BH ensures that the \emph{mean} of the FDP is controlled, while we ensure that the \emph{median} of the FDP is controlled.
The mean and the median of the FDP can be asymptotically equal in some settings where the dependencies among  the \emph{p}-values are not too strong \citep{neuvial2008asymptotic,ditzhaus2019variability}, but there is no general guarantee that they are similar \citep{romano2006stepup,schwartzman2011effect}. Especially under strong dependence, $mFDP\leq \gamma$  does not need to imply $\mathbb{E}(FDP)\leq \gamma$, while the converse does hold in many practical situations. Moreover, unlike mFDP control, FDR control always implies weak control of the family-wise error rate \citep[][Section 6.4]{romano2008formalized}. Note, however, that before applying any multiple testing method, we could first perform a global test, to enforce weak family-wise error rate control \citep{bernhard2004global}.

The most important advantage of our method over BH, is that it provides simultaneous 50\% confidence bounds for the FDP.  This allows simultaneous as well as post hoc inference, in the sense that $t\in\mathbb{T}$ can be chosen after seeing the data. Further, we can choose multiple values of  $t$ and obtain simultaneously valid statements on the FDP. Moreover, we can choose the target FDP $\gamma$ post hoc. With BH, such inference is not possible: if we choose $\gamma$ after seeing the data, then BH can become very anti-conservative. This is discussed in subsection \ref{secBHnotflexible}.

\subsection{Benjamini-Hochberg is not flexible} \label{secBHnotflexible}
The main advantage of the method that we will propose, is that it allows choosing one or several rejection thresholds or  target FDPs after seeing the data. This contrasts our method with BH. Indeed, when the target FDR $\alpha$ (or $\gamma$ in our notation) is chosen based on the data, then BH no longer guarantees that 
$\mathbb{E}(FDP)\leq\alpha$, conditional on the post hoc chosen $\alpha$.
Note that when testing a single hypothesis, choosing $\alpha$ post hoc is not generally valid either  
 \citep{hubbard2004alphabet,grunwald2022beyond}.
For simulations illustrating that BH is not valid post hoc, see the Supplementary Material. Another related result is Fig. 5 in \citet{katsevich2020simultaneous}, which illustrates based on simulations  that BH does not have a simultaneous interpretation.
We now provide some mathematical examples that prove that BH is not valid post hoc.

Suppose all $m$ hypotheses are true and the \emph{p}-values are mutually independent and uniformly  distributed on $(0,1]$. BH provides $m$ \emph{adjusted} \emph{p}-values and rejects all hypotheses with adjusted \emph{p}-values that are at most $\alpha$. Let  $p_{(1)}^{bh}$ denote the smallest adjusted \emph{p}-value. 
It is well known that if $\alpha\in[0,1]$ is prespecified and all \emph{p}-values are independent and uniform on $(0,1]$, then the probability that BH rejects any hypotheses is exactly $\alpha$ \citep{goeman2011multiple}.
BH rejects any hypotheses if and only if $p_{(1)}^{bh}\leq \alpha$. 
Thus, $p_{(1)}^{bh}$ is uniform on $(0,1]$.

As a simple example of a post hoc chosen $\alpha$, take  $\alpha:=p_{(1)}^{bh}$. We now show that we then no longer have   $\mathbb{E}(FDP/\alpha)\leq 1$. 
Since  $\alpha= p_{(1)}^{bh}$, we know that $\alpha$ is uniform on $(0,1]$.
By definition of $\alpha$, there is always at least one rejected hypothesis. Since all hypotheses are true, this means that we always have $FDP=1$.
 Consequently,
 $$\mathbb{E}(FDP/\alpha) = \mathbb{E}(1/\alpha) = \int_0^1 x^{-1} dx =log(x) \Big|_0^1 =\infty,$$
i.e., $\mathbb{E}(FDP/\alpha)$ is completely out of control.

 Of course, we have considered an extreme situation, where $\alpha$ can take any value. We now consider a less extreme situation, where we only allow $\alpha$ to take two values, say, $a_1$ and $a_2$, 
with $0<a_1\leq a_2<1$. Specifically, we define $\alpha$ to be $a_1$ if $ p_{(1)}^{bh}\leq a_1$ and otherwise we take $\alpha=a_2$. This mimics the psychology of a researcher who uses $a_2$ as a default value for $\alpha$ but takes $\alpha$ to be $a_1$ if this still leads to at least one rejection.
Note that if $ p_{(1)}^{bh}> a_2$, then we reject nothing, so that $FDP=0$.
Thus, with this definition of $\alpha$, we have
 $$\mathbb{E}(FDP/\alpha) =$$
 $$ \mathbb{P}\Big(p_{(1)}^{bh}\leq a_1\Big)\mathbb{E}\Big(FDP/a_1 \Big| p_{(1)}^{bh}\leq a_1\Big) +  \mathbb{P}\Big(a_1 < p_{(1)}^{bh}\leq a_2\Big) \mathbb{E}\Big(FDP/\alpha \Big| a_1< p_{(1)}^{bh}\leq a_2\Big)=$$
 $$
 a_1\mathbb{E}\Big(1/\alpha \Big| p_{(1)}^{bh}\leq a_1\Big) +  (a_2-a_1) \mathbb{E}\Big(1/\alpha \Big| a_1< p_{(1)}^{bh}\leq a_2\Big)=
 $$
 $$  a_1/a_1   +  (a_2-a_1)/a_2  = 1   +  (a_2-a_1)/a_2,$$ 
 which always exceeds 1, except if $a_1=a_2$.
 As an example, take $a_1=0.05$ and $a_2=0.1$, which are values often used in practice. This defines a rather limited set of allowed values for $\alpha$. Nevertheless, we find that $\mathbb{E}(FDP/\alpha)=1.5$, which is already much larger than 1. The reader can check that if we allow $\alpha$ to take more than two values, then $\mathbb{E}(FDP/\alpha)$ can become huge. Indeed, if we allow $\alpha$ to take any value in $(0,1]$, then $\mathbb{E}(FDP/\alpha)$ can become infinity, as we saw in the previous example where $\alpha= p_{(1)}^{bh}$.

  These examples show that if $\alpha$ depends on the data, then marginally we often  have 
 $\mathbb{E}(FDP/\alpha)> 1$. 
 This means in particular that conditional on  $\alpha$ taking a certain value, we do not generally have  $\mathbb{E}(FDP)\leq  \alpha$. A simulation study that illustrates this point in various other settings, is in the Supplementary Material.


 



\subsection{Simultaneous bounds for the FDP} \label{secsimbounds}

Let $\mathbb{N}$ be the set of natural numbers. We call a function $B: \mathbb{T}\rightarrow \mathbb{N}$ a \emph{confidence envelope} is it satisfies inequality \eqref{envelope}  \citep[cf.][]{hemerik2019permutation}. We restrict ourselves to such 50\% confidence envelopes and do not consider e.g. 95\% confidence envelopes.
Let $\mathbb{B}$ be a set of maps $\mathbb{T}\rightarrow \mathbb{N}$. Assume that $\mathbb{B}$ is monotone, in the sense for all $B,B'\in \mathbb{B}$, either $B\geq B'$ or $B'\geq B$. Here $B\geq B'$ means that $B(t)\geq B'(t)$ for all $t\in \mathbb{T}$. We call $\mathbb{B}$ the  \emph{family of candidate envelopes}  \citep[cf.][]{hemerik2019permutation}.

We will obtain a confidence envelope by picking the smallest $B\in \mathbb{B}$ for which $B(t)\geq \overline{V}(t)$ for all $t\in \mathbb{T}$. We call this envelope $\tilde{B}$:
$$\tilde{B}=\tilde{B}(p)=\min\Big\{B\in \mathbb{B}: \bigcap_{t\in \mathbb{T}} \big\{B(t)\geq\overline{V}(t)\big\} \Big\},$$
If $r$ is a vector containing, say, $l_r$ \emph{p}-values, then we write $\R(r,t)=\{1\leq i \leq l_r: r_i<t\}$,  to make the dependence on the \emph{p}-values explicit. Analogously we define ${V}(r,t)$,  $\overline{V}(r,t)$ and $\tilde{B}(r)$. We use the convention that  $\R(t)=\R(p,t)$, ${V}(t)={V}(p,t)$, $\overline{V}(t)=\overline{V}(p,t)$ and $\tilde{B}=\tilde{B}(p)$.


We will only require the following assumption.

\begin{assumption} \label{assB}
The following holds:
\begin{equation} \label{asuniform}
   \mathbb{P}\big\{\tilde{B}(p)\geq \tilde{B}(1-q)\big\}\geq 0.5, 
\end{equation}
\end{assumption}
Due to the monotonicity of the set $\mathbb{B}$, we always have either  $\tilde{B}(q)< \tilde{B}(1-q)$ or  $\tilde{B}(q)\geq  \tilde{B}(1-q)$. 
If the latter inequality has the largest probability, then Assumption \ref{assB} is always satisfied, since $\tilde{B}(p)\geq \tilde{B}(q)$.
 Assumption \ref{assB} is a generalization of Assumption \ref{as1}, in the sense that if $\mathbb{T}$ is equal to the singleton $\{t\}$ then Assumptions \ref{as1} and \ref{assB} will  coincide, for most reasonable choices of $\mathbb{B}$ (e.g. for $\mathbb{B}$ as in Section \ref{secspecific}).

 Assumption \ref{assB} always holds if property \eqref{eqsymm} is satisfied, regardless of our choice of $\mathbb{B}$.
Indeed, if \eqref{eqsymm} holds, we have 
$\mathbb{P}\big\{\tilde{B}(p)\geq \tilde{B}(1-q)\big\} \geq \mathbb{P}\big\{\tilde{B}(q)\geq \tilde{B}(1-q)\big\} = \mathbb{P}\big\{\tilde{B}(1-q)\geq \tilde{B}(q)\big\}$. Since the latter two probabilities are equal, they are both at least $0.5$, so that Assumption \ref{assB} is satisfied.
Moreover, property  \eqref{eqsymm} is not necessary for Assumption \ref{assB} to hold, as confirmed by our simulations.

Let $[\cdot]^+$ be the positive part function.
The following theorem 
states that $\tilde{B}$ provides simultaneously valid $50\%$-confidence bounds. 

\begin{theorem} \label{thmenvelope}
Suppose Assumption \ref{assB} holds.
Then the function $\tilde{B}$ is a confidence envelope, i.e.,
$$ \mathbb{P}\Big[\bigcap_{t\in \mathbb{T}} \{V(t)\leq \tilde{B}(t)\}\Big]\geq 0.5$$
and
$$ \mathbb{P}\Big[\bigcap_{t\in \mathbb{T}} \{FDP(t)\leq \tilde{B}(t)/R(t)\}\Big]\geq 0.5.$$

In addition, $\tilde{B}': \mathbb{T} \rightarrow \mathbb{N}$ defined by 
$$\tilde{B}'(t) = R(t) -\max\{[R(l)-\tilde{B}(l)]^+: l\in \mathbb{T}, l\leq t)\} ,$$
which satisfies $\tilde{B}'\leq  \tilde{B}$, is also a confidence envelope and potentially improves $\tilde{B}$.
\end{theorem}

\emph{Discussion of the proof.} The proof is the Supplementary Material, but here we will give the intuition. 
First of all, note that $\overline{V}(t)$  is a $50\%$ confidence bound for $V(t)$, but not simultaneously over all $t$.
The reason is that if multiple events have probability $0.5$, then the probability that all events happen is usually smaller than $0.5$. For example, if for $t_1$, $t_2\in(0,1)$ we have 
$ \mathbb{P}\big(V(t_j)\leq \overline{V}(t_j)\big)\geq 0.5$ for $j=1$ and $j=2$, then we do not generally have  $ \mathbb{P}\big(V(t_1)\leq \overline{V}(t_1) \text{ and } V(t_2)\leq \overline{V}(t_2)\big)\geq 0.5$.
To get a simultaneous bound for $V(t)$, we usually need a stricter requirement, i.e., it is not sufficient to simply define $\tilde{B}(t) = \overline{V}(t)$. 
In the proof, we note that if $\tilde{B}(p)\geq \tilde{B}(1-q)$, then $\tilde{B}(t) \geq V(t)$ for all $t$. It thus  follows from  Assumption 2 that our $\tilde{B}$ is a confidence envelope.
That  $\tilde{B}$ is chosen from a fixed, monotone family is not directly used in the proof.
However, if $\tilde{B}$ is chosen from such  a family, then $\tilde{B}(p)\geq \tilde{B}(q)$ and it follows that if  \eqref{eqsymm} holds, then
 Assumption 2  is satisfied. 
Thus, that  $\tilde{B}$ is chosen from a fixed, monotone family  makes Assumption 2 reasonable. It also allows defining $\tilde{B}$ as a simple minimum.


In the rest of this subsection, we  provide an extention of the bounds $\tilde{B}'(t)$  and a result on admissibility. 
It turns out that $\tilde{B}'$ coincides with an envelope obtained through a novel \emph{closed testing}-based procedure, in the sense of  \citet{goeman2011multiple} and \citet{goeman2021only}. 
This novel procedure provides a  $50\%$ confidence bound for the number of true hypotheses in $I$, for \emph{every} subset $I\subseteq \{1,...,m\}$. These bounds are all simultaneously valid with probability at least $50\%$. We denote these bounds by $\overline{B}(I)$.

\begin{theorem} \label{thmctmethod}
 Write $\mathcal{M}=\{I\subseteq\{1,...,m\}:I\neq\emptyset\}.$
For every $I\in\mathcal{M}$ and $t\in\mathbb{T}$, define $R_I(t):=|\R(t)\cap I|=|\{i\in I: p_i\leq t\}|$. Write 
\begin{equation}  \label{formctbound}
\overline{B}(I):= \max\{|A|: \emptyset\neq A\subseteq I\text{ and }\forall t\in\mathbb{T}: R_A(t)\leq \tilde{B}'(t)\}. 
\end{equation} 
where the maximum of an empty set is interpreted as 0.

Assume $\mathbb{T}\subseteq[0,1/2)$ and $ \mathbb{P}\big\{\tilde{B}(q)\geq  \tilde{B}(1-q)\big\}\geq 0.5 $.
Then $$\mathbb{P}\Big[ \bigcap_{I\in\mathcal{M}} \big\{ |\N\cap I|\leq    \overline{B}(I) \big\}   \Big]\geq0.5,$$
i.e., the $\overline{B}(I)$ are simultaneous $50\%$ confidence bounds for the number of true hypotheses in $I$. In particular, the function $\mathbb{T}\rightarrow\mathbb{N}$ defined by $t\mapsto \overline{B}(\R(t))$ is a confidence envelope.

\end{theorem}

 If the \emph{local tests}  discussed in the proof of Theorem \ref{thmctmethod} are admissible, then the method of Theorem \ref{thmctmethod} is admissible, in the sense of  Theorem 3 in \citet{goeman2021only}. 
 The local tests will usually be admissible when $\mathbb{B}$ is any reasonable family, for example the family considered in Section \ref{secspecific}. 
By Theorem \ref{thmequivct} below, if the procedure of Theorem \ref{thmctmethod} is admissible, then the envelope $\tilde{B}'(t)$ from Theorem \ref{thmenvelope} is also admissible.
 


\begin{theorem} \label{thmequivct}
For every $t\in\mathbb{T}$, the bound $\tilde{B}'(t)$ from Theorem \ref{thmenvelope} is equal to the  bound $\overline{B}(\R(t))$ from Theorem \ref{thmctmethod}.
Moreover, if the procedure from Theorem \ref{thmctmethod} that provides bounds for all $I\in \mathcal{M}$ is admissible, then the envelope $\tilde{B}'$ is also admissible. Here admissibility of  $\tilde{B}'$ means that there exists no envelope $B:\mathbb{T}\rightarrow \mathbb{N}$ such that $B(t)\leq \tilde{B}'(t)$ for all $t\in\mathbb{T}$ and such that $\mathbb{P}\{\exists t\in\mathbb{T}:B(t)< \tilde{B}'(t)\}>0.$
\end{theorem}

The admissibility property of our method contrasts it with BH. The latter method is not admissible, since it is uniformly improved by the method in \citet{solari2017minimally}, of which it is not known whether it is admissible.
In the rest of this paper we will only focus on bounds for rejected sets of the form $\R(t)=\{1\leq i \leq m: \text{ } p_i\leq t\}$, as constructed in Theorem \ref{thmenvelope}.

\subsection{A default mFDP envelope} \label{secspecific}

The envelope $\tilde{B}$ depends on a general family $\mathbb{B}$ of candidate confidence bounds. The choice of this family can have a large influence on the bounds obtained  \citep[cf.][]{hemerik2019permutation}. An important question is thus how to choose this set $\mathbb{B}$ in a suitable way. Typically we want $\mathbb{B}$ to contain at least one function $B$ that is a tight  upper envelope of the function $t\mapsto \overline{V}(t)$. Note that between $t=0$ and, say, $t=0.5$, the function $\overline{V}(t)$ tends to be roughly linear in $t$, at least under  independence. Thus, it can make sense to also take the candidate envelopes $B\in \mathbb{B}$ to be roughly linear. Also, giving them a small positive intercept  will often be useful to avoid that $\tilde{B}$ is too sensitive to \emph{p}-values near 1.

Further, it is usually suitable to take $\mathbb{T}=[s_1,s_2]$, where $s_1\geq 0$ is the smallest threshold of interest and $s_2<1$ is the largest threshold of interest. 
Based on these considerations, we propose to use the following default family $\mathbb{B}$ of candidate functions:
\begin{equation} \label{examplefamily}
\mathbb{B}=\{B^{\kappa}: \kappa\in (0,\infty]\},
\end{equation}
with 
$$ B^{\kappa}(t)= |\{1\leq i \leq m: i\kappa - c \leq t\}|= \Big\lfloor \frac{t+c}{\kappa}\Big\rfloor.$$
Here, $c\geq 0$ is a pre-specified small constant. The discrete function $B^{\kappa}$ is roughly linear in $t$ and has slope $1/\kappa$. 

The choice of $c$  influences the slope and intercept of $B^{\kappa}$ and hence of the resulting envelope $\tilde{B}$. Taking $c$ to be $0$ or very small tends to lead to tighther bounds $\tilde{B}(t)$ for very small $t$, while taking $c$ a bit larger tends to lead to tighter bounds for larger $t$. We found in simulations that taking $c=1/(2m)$ usually gave good overall power.

If  we take  $\mathbb{B}$ as in expression \eqref{examplefamily}, then the confidence envelope becomes
\begin{equation} \label{Bkmax}
\tilde{B}=B^{\kappa_{\text{max}}} \text{, where } \quad
\kappa_{\text{max}}=\max\Big\{\kappa\in (0,\infty]: \bigcap_{t\in \mathbb{T}} \big\{ B^{\kappa}(t) \geq  \overline{V}(t) \big\} \Big\}.
\end{equation}
For computer programming this method, a useful equivalent formulation is the following, if $\mathbb{T}$ is an interval.

\begin{proposition} \label{propcomputekappa}
Suppose $\mathbb{T}$ is of the form $[s_1,s_2]$, with $0\leq s_1<s_2\leq 1$.
We then have  
\begin{equation} \label{reformulation_kappamin}
\kappa_{\text{max}}=\kappa_0\wedge\min\Big\{\kappa_i: 1\leq i\leq m \text{ and } 1-p_i\in  \mathbb{T} \Big\},
\end{equation}
where 
$$ \kappa_0 = \frac{s_1+c}{\overline{V}(s_1)}= \frac{s_1+c}{|\{1\leq j \leq m:\text{ } p_j\geq 1-s_1\}|}$$
and for $1\leq i \leq m$
$$ \kappa_i = \frac{1-p_i+c}{\overline{V}(1-p_i)} = \frac{1-p_i+c}{|\{1\leq j \leq m:\text{ } p_j\geq p_i\}|}. $$
(If the denominator is 0, the expression is interpreted as $\infty$.)
\end{proposition}

Note that we can sometimes straightforwardly  improve the envelope $B^{\kappa_{\text{max}}}$ by using the second part of Theorem \ref{thmenvelope}. In Example \ref{rex2}, we continue the running example and compute simultaneous mFDP bounds. Figure \ref{fig:bounds} shows the confidence envelope and Figure \ref{fig:constr} illustrates how the envelope was determined.

\begin{example}[Running example, part 2: Confidence envelopes.] \label{rex2}
We continue on Example \ref{rex1} by computing confidence envelopes, i.e., simultaneous $50\%$-confidence upper bounds for $V(t)$, the number of false positives, which depends on the threshold $t$. We took $\mathbb{T}=[0,0.2]$ and defined $\tilde{B}$ as in  \eqref{Bkmax}. We computed $\tilde{B}$ for both $c=0$ and $c=2/m=0.004$. These choices for $c$ were somewhat arbitrary. 
The number of rejections $R(t)$, as well as the bounds $\tilde{B}(t)$ for both values of $c$, are plotted in Figure \ref{fig:bounds}.
The construction of the confidence envelopes $\tilde{B}$ is illustrated in Figure \ref{fig:constr}.

The figure shows that as expected, near $t=0$, the number of rejections increases quickly with $t$. The reason is that there were many p-values near 0, as seen in Figure \ref{fig:Storey}. By definition \eqref{Bkmax}, the bounds $\tilde{B}(t)$ are roughly linear in $t$ and we see this in the figures. We also see that for this specific dataset, the bound $\tilde{B}$ depends strongly on $c$: for $c=0.004$, the bound $\tilde{B}(t)$ is lower than for $c=0$ if $t$ is close to $0$, but much higher otherwise. For most values of $t\in[0,1]$ the envelope for $c=0.004$ is better, i.e. lower, than the envelope for $c=0$.  On the other hand,  the smallest cutoffs are often most relevant.
Finally, we remark that the bounds in the figures can be somewhat improved using the last part of Theorem  \ref{thmenvelope}.
This improvement was used to obtain Figure \ref{fig:FDPbounds}, where simultaneous $50\%$ confidence bounds for $FDP(t)$ are shown.
\end{example}

\begin{figure}[ht]   
\centering
  \includegraphics[width=0.8\linewidth]{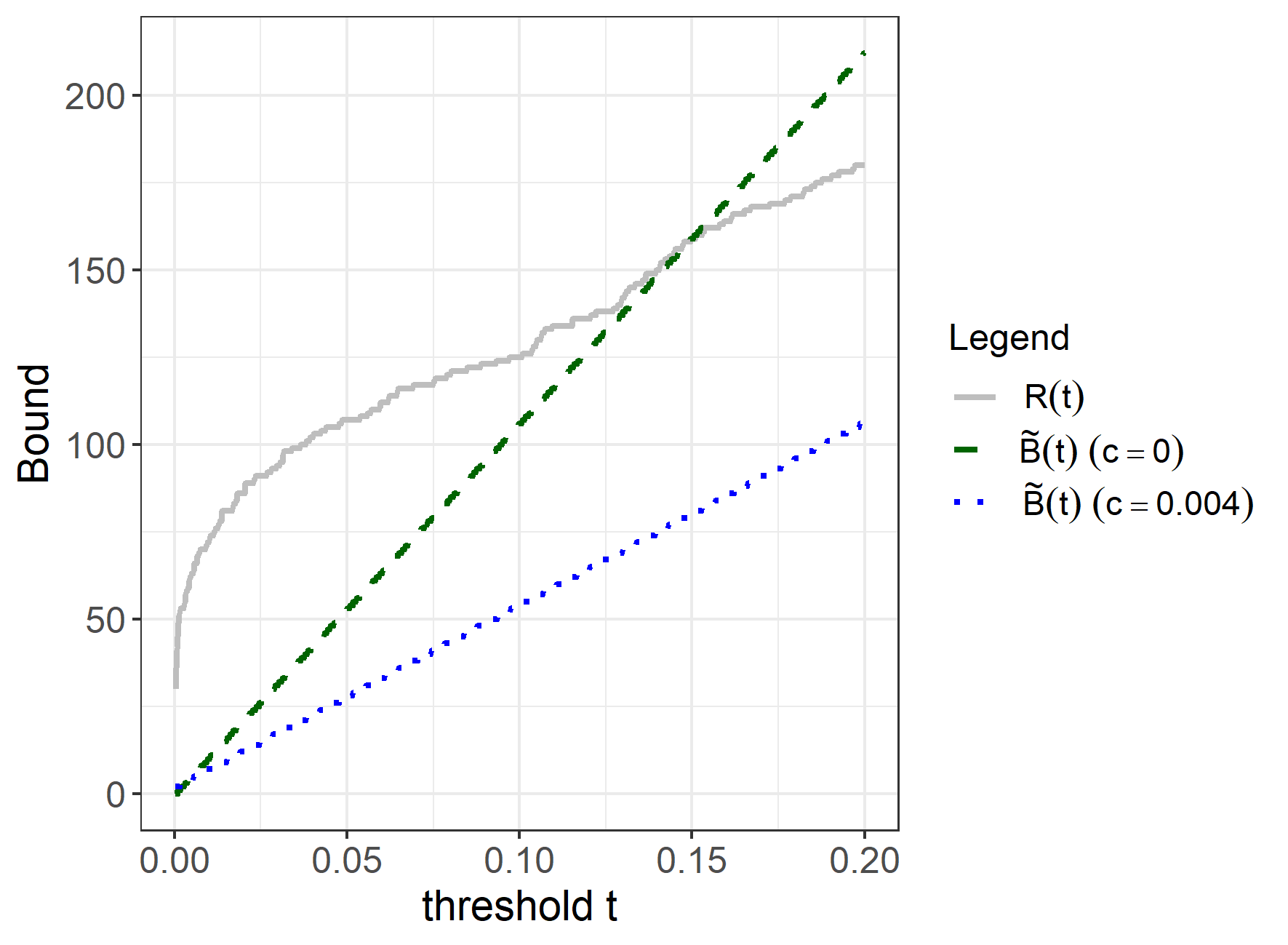}
  \caption{Graph showing the number of rejections and two confidence envelopes for the running example. The solid line shows the number of rejections, which depends on the rejection threshold $t$. The other lines are two confidence envelopes $\tilde{B}$. These are simultaneous $50\%$-confidence upper bounds for the number of false positives $V(t)$. The intercept and slope of $\tilde{B}$ depend on the user-specified constant $c$. Note that for $c=0$, the intercept is slightly smaller than for $c=0.004$. Indeed, the intercepts are 0 and 2 respectively.}    
  \label{fig:bounds}
\end{figure}

\begin{figure}[ht]  
\centering
  \includegraphics[width=0.8\linewidth]{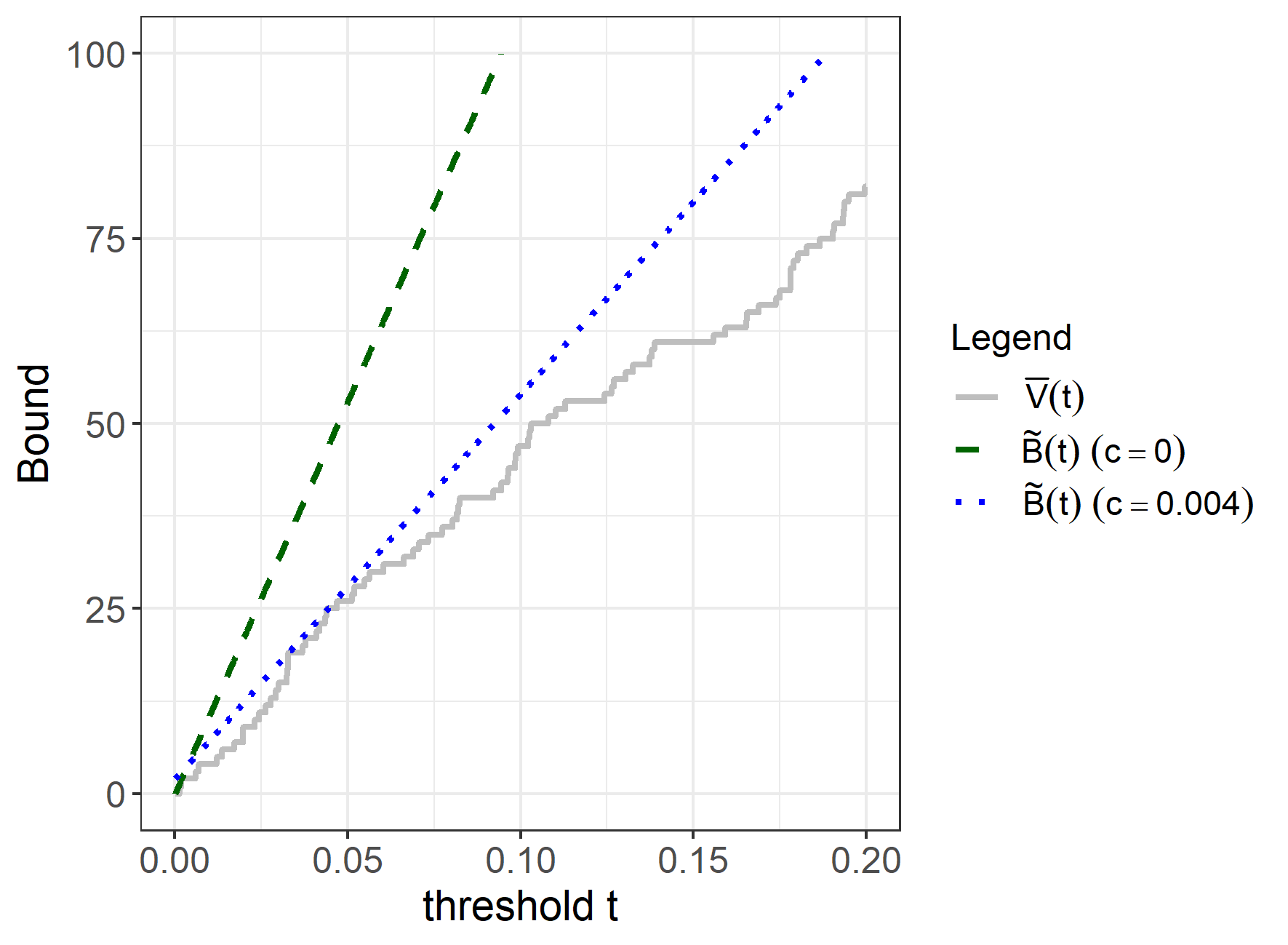}
  \caption{Illustration of the construction of the confidence envelope for the running example. For every rejection threshold $t$,  $\overline{V}(t)$ is a $50\%$ confidence upper bound for the number of false positives, $V(t)$.
  The confidence envelope $\tilde{B}(t)$ is constructed in such a way that it lies above the pointwise bound $\overline{V}(t)$ for all $t\in\mathbb{T}$. Due to this construction, the bounds $\tilde{B}(t)$ are simultaneous $50\%$-confidence bounds for $V(t)$. The intercept and slope of $\tilde{B}$ are influenced by the choice of $c$.}      
  \label{fig:constr}
\end{figure}

\begin{figure}[ht]  
\centering
  \includegraphics[width=0.8\linewidth]{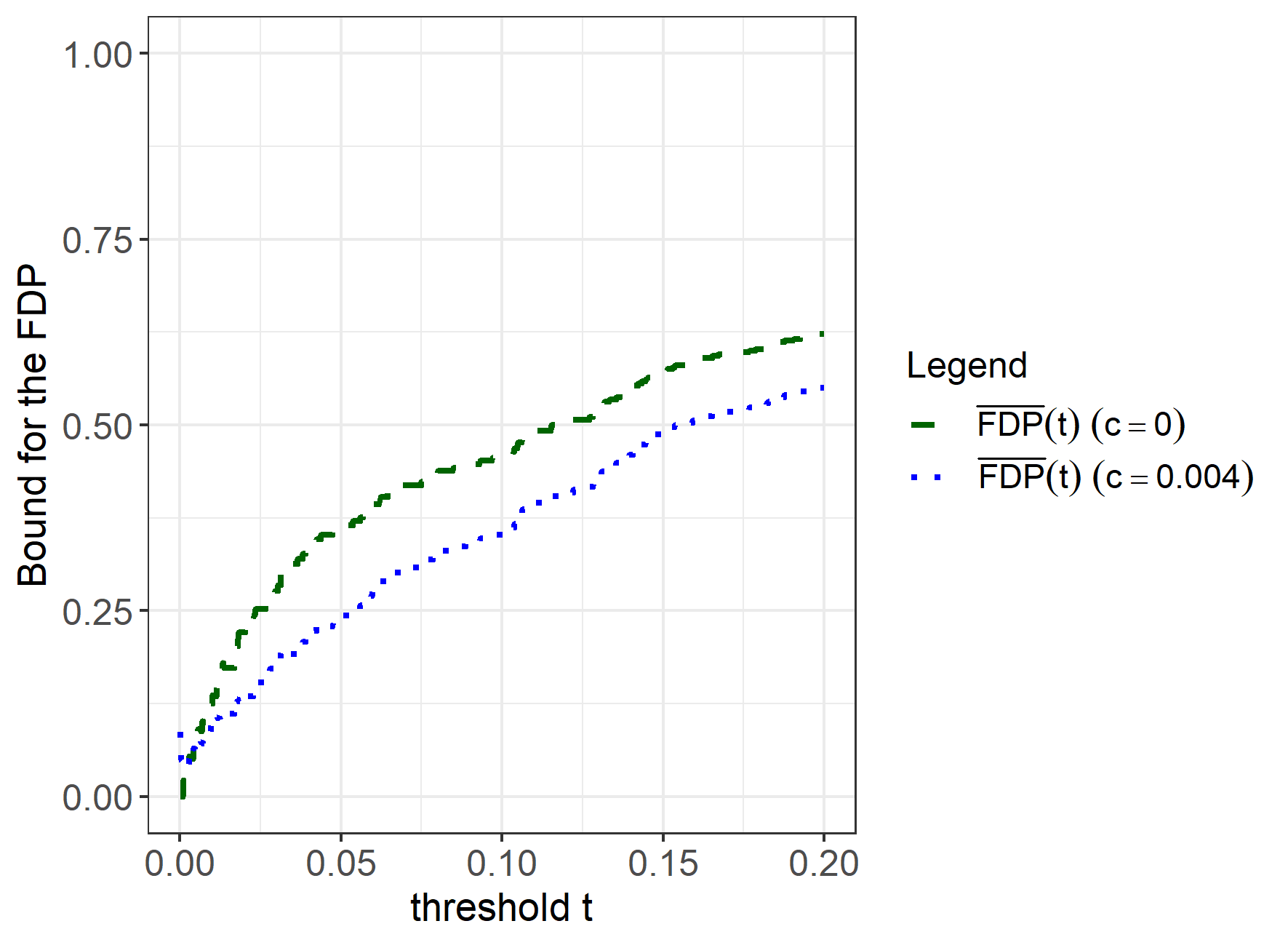}
  \caption{For two values of $c$, simultaneous $50\%$ confidence upper bounds for $FDP(t)$ are shown. Here $\overline{FDP}(t):=\tilde{B}(t)/R(t)$. Note that if $c=0.004$, the bound is larger than zero at $t=0$. The reason is that $\tilde{B}(0)>0$ for this value of $c$. Roughly speaking, the bound $\overline{FDP}(t)$ then decreases for a while, before it starts to increase. Note that if $c=0$, then the bound starts at zero and increases from there.}
  \label{fig:FDPbounds}
\end{figure}

\subsection{Controlling the median of the FDP} \label{seccontr}
Consider $\gamma\in [0,1]$.
As discussed in section \ref{secoverv}, we can use any 
confidence envelope $B$  to guarantee  that $\mathbb{P}(FDP\leq \gamma)\geq 0.5.$ 
In other words, we can control the mFDP. Note that by $mFDP$ we mean the median of the distribution that the FDP has, conditional on the data and conditional on $\gamma$, which can be chosen after seeing the data. 
This is stated in the following Theorem. (The maximum of an empty set is taken to be 0.)  

\begin{theorem} \label{fdpcontrol}
Let $B:\mathbb{T}\rightarrow \mathbb{N}$ be a confidence envelope, for example $\tilde{B}$.
Let the target FDP $\gamma\in[0,1]$ be  freely chosen based on the data.
Define 
$$t_{\text{max}}=t_{\text{max}}(B,\gamma)=\max\{p_i: \text{ there is a } t\in \mathbb{T}\cap[p_i,1]: B(t)/R(t) \leq \gamma\}.$$ %
Reject all hypotheses with p-values at most $t_{\text{max}}$ and denote the FDP by $FDP_\gamma$. Then with probability 0.5 the FDP is at most $\gamma$, i.e.,
\begin{equation} \label{FDPcontrol}
\mathbb{P}\big\{FDP_\gamma\leq \gamma)\big\}\geq 0.5.
\end{equation}
In fact we have
\begin{equation} \label{eq:uniformovergamma}
 \mathbb{P}( \bigcap_{\gamma\in[0,1]} FDP_{\gamma}\leq\gamma )\geq 0.5,
\end{equation}
i.e., the procedure offers mFDP control simultaneously over all $\gamma\in[0,1]$.

\end{theorem}

In other words, if we reject all hypotheses with \emph{p}-values that are at most $t_{\text{max}}$, then a median unbiased estimate of the FDP is $\gamma$. This follows directly from the fact that the estimates $\overline{FDP}(t)$, $t\in \mathbb{T}$, are simultaneously valid $50\%$-confidence upper bounds, by inequality \eqref{envelope}. Inequality \eqref{FDPcontrol}  holds despite the fact that $\gamma$ can depend on the data. In fact, with probability at least $50\%$, $FDP_{\gamma}\leq\gamma$ simultaneosly over all $\gamma\in[0,1]$.
 This contrasts our method with many other procedures, which require considering only one rejection criterion, which moreover needs to be chosen in advance \citep{benjamini1995controlling, van2004augmentation,lehmann2005generalizations, romano2007control,guo2007generalized,roquain2011type, neuvial2008asymptotic, guo2014further,delattre2015new, ditzhaus2019variability, dohler2020controlling, basu2021empirical, mie2022}. In Example \ref{rex3} we continue the running example and apply our mFDP control method.

\begin{example}[Running example, part 3: Controlling the mFDP.] \label{rex3}
We continue on Example \ref{rex2}. Take $\gamma=0.05$ and consider the confidence envelope $\tilde{B}$ discussed in Example \ref{rex2}. To find a rejection threshold  $t_{\text{max}}$ for which we can ensure $mFDP\leq\gamma$, we use Theorem \ref{fdpcontrol}. It computes $t_{\text{max}}$ as the largest $t$ for which the estimate in Figure \ref{fig:FDPbounds} is at most $\gamma$.

Recall that in Example \ref{rex2}, we computed bounds $\tilde{B}(t)$ for both  $c=0$ and $c=0.004$.
For $c=0$, we now find $t_{\text{max}}=0.002709$, which is the 54-th smallest \emph{p}-value. Thus, we can reject 54 hypotheses. 
More precisely, if we reject the 54 smallest \emph{p}-values, we know that the mFDP is below $\gamma=0.05$.
Note that $t_{\text{max}}$ is about 27 times higher than the Bonferroni threshold $0.05/500=0.0001$.

If $c=0.004$ then $t_{\text{max}}=0.001660$, so that we can only reject 53 hypotheses. The reason why $t_{\text{max}}$ is lower if $c=0.004$, it that  for small values of $t$, the bound $\tilde{B}(t)$ is higher for $c=0.004$ than for $c=0$. We saw this in Figure \ref{fig:bounds}.

Note that it is allowed to change $\gamma$ after looking at the data. For instance, if we decrease $\gamma$ to 0.01, we reject 44 hypotheses if $c=0$ and we reject no hypotheses for $c=0.004$.

\end{example}

\subsection{Adjusted \emph{p}-values for mFDP control} \label{secadjp}
Adjusted \emph{p}-values can be a useful tool in multiple testing. They are defined as the smallest level, e.g. the smallest $\gamma$, at which the multiple testing procedure would reject the hypothesis.
Adjusted \emph{p}-values can be problematic in the context of e.g. FDR control and ours. The reason is that the adjusted \emph{p}-value does not have an independent meaning and can easily be misinterpreted when taken out of context \citep[][\S5.4]{goeman2014multiple}.
Moreover, an mFDP-adjusted \emph{p}-value could be $0$, which also shows that the interpretation is very different than for real \emph{p}-values, which cannot be $0$.
Nevertheless, in our context, adjusted \emph{p}-values are quite useful, because, once computed, they allow checking quickly which hypotheses are rejected for various $\gamma$.

Let $B$ be a confidence envelope and $1\leq i \leq m$.  As discussed in Section \ref{seccontr}, $B$ defines an mFDP controlling procedure.   The mFDP adjusted \emph{p}-value for $H_i$ is the largest $\gamma\in[0,1]$ for which $H_i$ is still rejected by the mFDP controlling procedure. 
Consequently, if we reject all hypotheses $H_i$ with $p_i^{\text{ad}}\leq \gamma$, then $mFDP\leq \gamma$.

\begin{proposition} \label{propadaptedp}
Let $1\leq i \leq m$.
Then the value 
\begin{equation} \label{defpad}
p_i^{\text{ad}}:=\min\{B(t)/R(t) : t\in\mathbb{T}\cap{[p_i,1]}  \},
\end{equation}
is an  mFDP-adjusted \emph{p}-value for $H_i$, i.e., if we reject all hypotheses $H_i$ with $ p_i^{\text{ad}}\leq \gamma $,  then $\mathbb{P}(FDP_\gamma \leq \gamma)\geq 0.5$. Here $\gamma$ may be  chosen based on the data. In fact, inequality \eqref{eq:uniformovergamma} holds.
We take the minimum of an empty set to be  $\infty$.

\end{proposition}

Suppose $\mathbb{T}$, the set of rejection thresholds of interest,  is of the form $[s_1,s_2]$.
Then we have the following useful reformulation of Proposition \ref{propadaptedp}.

\begin{proposition} \label{rewriteadaptedp}
Suppose $\mathbb{T}$ is of the form $[s_1,s_2]$, with $0\leq s_1 < s_2\leq 1$.
For each $1\leq i \leq m$ with $p_i\leq s_2$, the adjusted \emph{p}-value defined above is then

\begin{equation} \label{pad2}
p_i^{\text{ad}}=\min\Big\{B(t)/R(t) :  t\in \big[\max\{s_1,p_i\},s_2\big] \cap\big\{s_1,p_1,p_2,...,p_m\big\}      \Big\}.
\end{equation}
\end{proposition}

Note that given the data, the adjusted \emph{p}-value is non-decreasing function of the unadjusted \emph{p}-value. As a consequence of this and Proposition \ref{rewriteadaptedp}, if $\mathbb{T}$ is of the form $[s_1,s_2]$, we can use Algorithm \ref{a:adjustedp} to efficiently compute the mFDP adjusted \emph{p}-values. 
The algorithm takes the $m$ sorted \emph{p}-values, $p_{(1)},...,p_{(m)}$, as input and returns the corresponding sorted adjusted \emph{p}-values.

The idea of the  algorithm is to start with computing the largest adjusted \emph{p}-value(s), continue with the second largest one and so on. The algorithm also uses the fact that if $p_{(i)}>s_2$, then $p_{(i)}^{\text{ad}}=\infty$. It further uses the fact that all hypotheses with unadjusted \emph{p}-values below $s_1$ have the same adjusted \emph{p}-value. Adjusted \emph{p}-values can be easily computed using the \verb|R| package \verb|mFDP|.

\begin{algorithm}[ht!] 
\caption{Algorithm for computing the mFDP adjusted \emph{p}-values if $\mathbb{T}=[s_1,s_2]$}
\begin{algorithmic}  \label{a:adjustedp}
\STATE $r \gets |\{1\leq i \leq m: \text{ }p_i\leq s_2\}|$.
\IF{$r<m$}
\STATE $p_{(r+1)}^{\text{ad}},...,p_{(m)}^{\text{ad}}\gets \infty$.
\ENDIF
\IF{$r>0$}
\IF{$s_1\leq p_{(r)}$}
\STATE $p_{(r)}^{\text{ad}} \gets B(p_{(r)})/R(p_{(r)})$
\ELSE 
\STATE $p_{(r)}^{\text{ad}} \gets B(s_1)/R(s_1)$
\ENDIF
\STATE $l \gets r-1$
\WHILE{$l>0$ AND $p_{(l)}\geq s_1$}
\STATE  $p_{(l)}^{\text{ad}}=\min\{p_{(l+1)}^{\text{ad}},  B(p_{(l)})/R(p_{(l)}) \}$
\STATE $l \gets l-1$
\ENDWHILE
\IF{$l>0$}
\STATE $p_{(1)}^{\text{ad}},...,p_{(l)}^{\text{ad}} \gets \min\{p_{(l+1)}^{\text{ad}},  B(s_1)/R(s_1)\}$ 
\ENDIF
\ENDIF
\RETURN $p_{(1)}^{\text{ad}},...,p_{(m)}^{\text{ad}}$
\end{algorithmic}
\end{algorithm}  \label{alg1}

\section{Simulations} \label{secsims}

We performed simulations to assess the error control, power and speed of our method, using R version 4.3 \citep{R}. We compared our approach  to three existing methods.
The first one is the method from \citet{goeman2019simultaneous}, which exploits Simes' inequality and closed testing and has proven admissibility, like our method.
That method is a special case of \citet{goeman2011multiple}, but \citet{goeman2019simultaneous} describes a faster algorithm, although its computational complexity is still not linear like that of ours.
 If one takes $\alpha=0.5$ in that method, then it provides flexible mFDP control, just like our method. The second method we compare with is the procedure for simultaneous FDP control from \citet{katsevich2020simultaneous}. Taking  $\alpha=0.5$ in that method gives flexible mFDP control, although the method does assume independence. Moreover, that method only has proven validity for $\alpha\leq0.31$, although the authors note that it is probably also valid for larger $\alpha$.
Finally, we compare our method with
 BH \citep{benjamini1995controlling,benjamini2001control}, which is the most popular method related to FDP control, although it does not control the mFDP but the FDR. Moreover, it requires choosing $\gamma$ (there called $\alpha$) before seeing the data. We also consider two so-called \emph{adaptive} FDR methods, that use an estimate of $\pi_0$ to gain power. 

In the simulations we considered $m=10^3$ or $10^4$ hypotheses. The \emph{p}-values were based on Z statistics, computed from simulated data with various dependence structures. The \emph{p}-values were two-sided, unless stated otherwise.
To add signal a number $\Delta$ was added to the first $(1-\pi_0)/m$ test statistics.
The following dependence structures of the test statistics were considered:
\begin{itemize}
\item independence (IN);
\item homogeneous positive correlations $\rho=0.5$ (HO);
\item five independent blocks, with positive dependence $\rho=0.8$ within blocks (BL);
\item 50 negatively dependent blocks (correlations -0.01),  with correlation 0.5 within blocks. The \emph{p}-values were right-sided, so that they were negatively correlated between blocks (NE).
\end{itemize}
Further, we varied $m$, $\pi_0$, and the signal $\Delta$.

We computed $\tilde{B}$ as in Section \ref{secspecific}.
We took $\mathbb{T}=[0,0.1]$, i.e. our bounds and mFDP-adjusted \emph{p}-values were simultaneously valid with respect to all thresholds $t$ in this interval. 
We took $c=1/(2m)$, as recommended in Section \ref{secspecific}.

We first assessed whether our method provided appropriate simultaneous mFDP control. We show simulation results  in Table \ref{table:error}. 
For each setting, the table shows the estimate of the probability 
$\mathbb{P}\{\text{for some } t\in\mathbb{T}\text{, }V(t)>\tilde{B}(t)\}$, which is identical to the probability that there is a $0<\gamma<1$ for which $FDP_{\gamma}$ exceeds $\gamma$.
Each estimate was based on $10^4$ repeated simulations. 

The table confirms the simultaneous control of our method. We see that the estimated error rate is  about $0.5$ under  independence if $\pi_0=1$. Indeed, the true error rate is then exactly 0.5. The reason is that then $p=q$ and the equality \eqref{eqsymm} then holds, so that the probability in Assumption \ref{assB} is exactly $0.5$. We see that for $\pi_0=0.95$, the error rate is also about $0.5$, rather than less.  This is because our method is rather adaptive, as mentioned in the Introduction. 
In the setting with negative dependence and $\pi_0=1$ and one-sided \emph{p}-values, the error rate is also exactly  0.5, again because  \eqref{eqsymm} then holds.
Note that in the other cases, the method was also valid. 

\begin{table}[!ht] \normalsize  
\caption{The error rate of our procedure, in various settings with $m=10^3$. The last column indicates the simulation-based  estimate of the probability that there is a $0<\gamma<1$ for which $FDP_{\gamma}$ exceeds $\gamma$. This probability should not be larger than $0.5$. For the settings with $\pi_0<1$, the signal for the false hypotheses was $\Delta=3$.} 
\begin{center}
    \begin{tabular}{ l l l l}    
\hline \\[-0.4cm]
$\pi_0$ &Setting & $\rho$  &  $\mathbb{P}(\text{error})$ \\ \hline 
1 & IN &$0$  \quad & .499 \\ 
1 & HO &$.2$  \quad & .334  \\ 
1 & HO&$.5$  \quad & .266  \\ 
1 & HO &$.9$  \quad & .330  \\ 
1 & BL &$.5$  \quad & .335  \\ 
1 & BL &$.9$  \quad & .351  \\ 
1 & NE &$-.01$  \quad & .500   \\ 
.95 & IN &$0$  \quad &  .498 \\ 
.95 & HO &$.2$  \quad & .336  \\ 
.95 & HO&$.5$  \quad & .266 \\ 
.95 & HO &$.9$  \quad & .327 \\ 
.95 & BL &$.5$  \quad &.338 \\ 
.95 & BL &$.9$  \quad & .343  \\ 
.95 & NE &$-.01$  \quad & .501  \\  \hline  
    \end{tabular}
\label{table:error}
\end{center}
\end{table}

Next, we assessed the power of our method by comparing it to that of the mentioned methods from \citet{goeman2019simultaneous} (``CT$+$Simes") and \citet{katsevich2020simultaneous} (``K\&R"). The power was defined as the average fraction of the false hypotheses that was rejected.
For three values of the target FDP $\gamma$ we estimated the power for the three methods.
The results are shown in Figure   \ref{fig:powers}, where $m=10^3$ and $\pi_0=0.9$.  
Note that overall K\&R performed least well among the three methods, especially for $\gamma=0.01$.  This may partly be due to the $+1$ in their formula for the bound for the number of false positives.
Further, for $\gamma=0.01$, CT$+$Simes had better power than our novel method.  However, as shown in the Supplementary material, for $m=10^4$ and $\pi_0=0.9$, our method was better than that method overall. Further, for $m=10^4$ and $\pi_0=0.5$ our method was clearly better than both competitors, as shown in Figure \ref{fig:powerextra}. This can be understood by noting that K\&R is not adaptive, i.e., it is conservative when $\pi_0$ is far from 1.

Further, our method was orders of magnitude faster than CT$+$Simes,  especially for large $m$. For example, in the setting of the first panel of Figure \ref{fig:powerextra}, our method took $1.7\cdot 10^{-2}$ seconds on average, while CT$+$Simes took 4.8 seconds on average. K\&R was the fastest with $8.6\cdot 10^{-4}$ seconds on average. The reason is that the bounds for $V(t)$ that that method provides, only depend on  $m$ and $t$ and not on the  data.

\begin{figure}[ht] 
\centering
  \includegraphics[width=\linewidth]{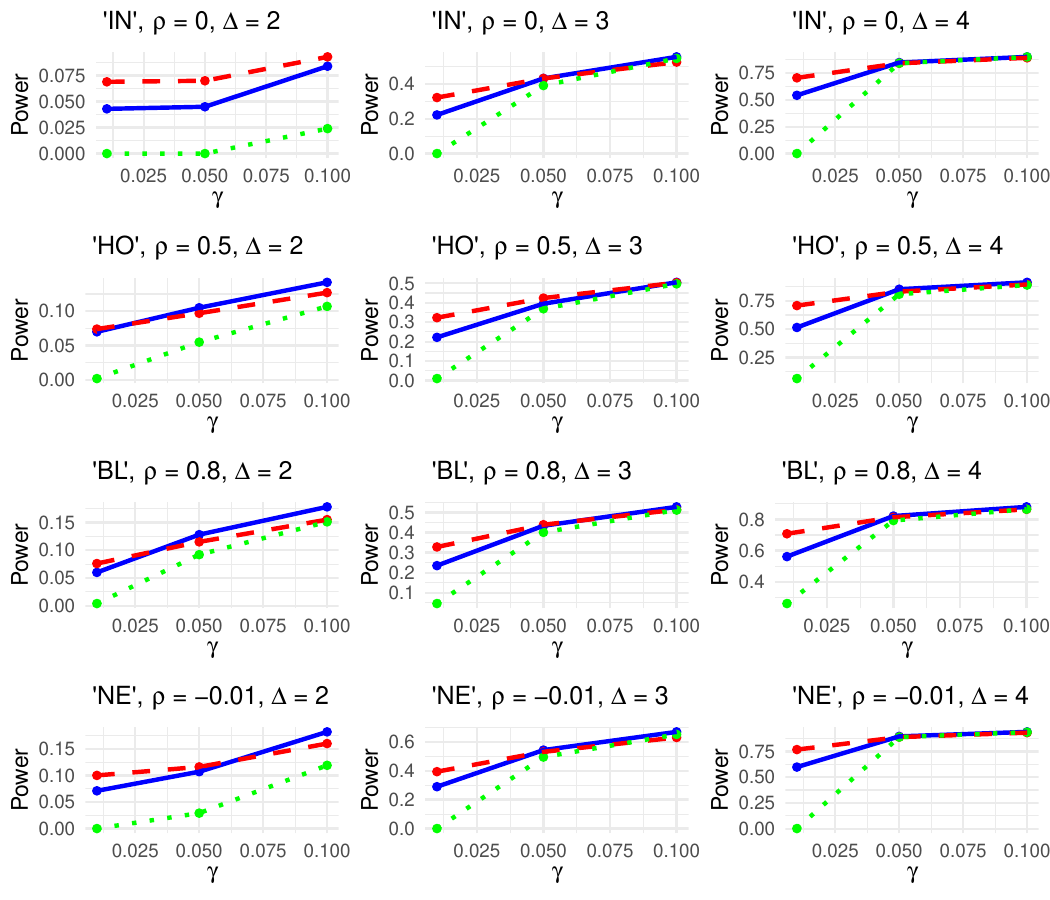}
  \caption{ The power of our novel method (solid lines), CT$+$Simes (dashed lines) and K\&R (dotted lines)  as depending on $\gamma$, for various settings with $m=10^3$ and $\pi_0=0.9$.  Each estimate is based on $10^4$ simulations.} \label{fig:powers}
\end{figure}

\begin{figure}[ht] 
\centering
  \includegraphics[width=\linewidth]{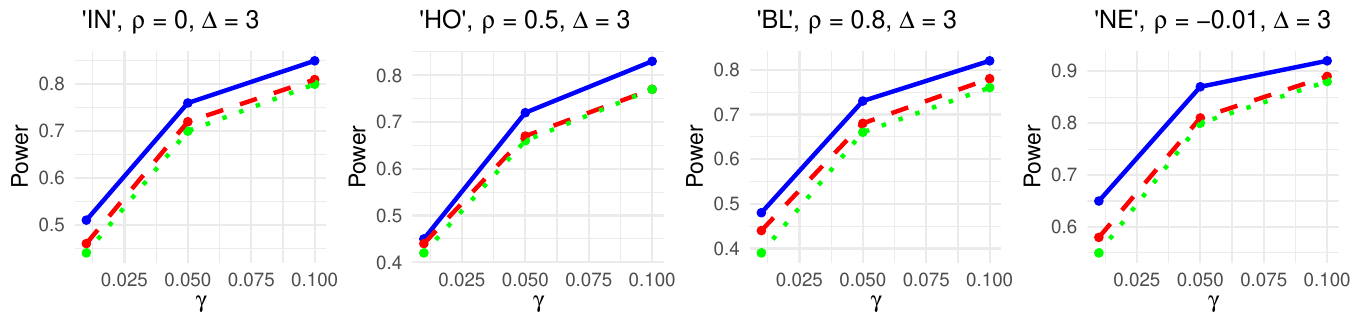}
  \caption{ The power of our novel method (solid lines), CT$+$Simes (dashed lines) and K\&R (dotted lines)  as depending on $\gamma$, for various settings with $m=10^4$ and $\pi_0=0.5$. Each estimate is based on $10^3$ simulations.} \label{fig:powerextra}
\end{figure}

Finally, for the same simulation settings, we computed the power of BH and two adaptive versions of BH. The results are in Table \ref{table:power3}. The first column shows the power of standard BH. The other columns show the power of two versions of the \emph{right-boundary procedure} of \citet{liang2012adaptive}, which makes BH more powerful by using an estimate of $\pi_0$. The first one (BH*) is the their original procedure based on Storey's estimator, $\hat{\pi}_0'$. The second one (BH**) is the same, but based on our novel estimator $\overline{\pi}_0'$.
Since BH and adaptive BH require choosing $\alpha$ beforehand, we only show the power for $\alpha=0.05$, i.e., $\gamma=0.05$.

\begin{table}[!ht] \normalsize  
\caption{ The power of BH and two adaptive versions of BH. The target FDR $\alpha$, i.e., $\gamma$, was $0.05$.  Each estimate in the table is based on $10^4$ simulations. } 
\begin{center}
    \begin{tabular}{ l l l l l l }    
\hline \\[-0.4cm]
  & &  & \multicolumn{3}{l}{\qquad  Method}\\ \cline{4-6} 
 Setting & $\rho$ & $\Delta$  &    BH \quad & BH* \quad & BH**\\ \hline 
  IN& 0       &$2$   & .058 & .062 &.062 \\ 
 IN& 0      &$3$   & .496 & .511 &.512\\ 
  IN& 0       &$4$   & .879 & .886&.886\\ 
   HO& .5     &$2$   & .099 & .125 &.122\\ 
   HO& .5      &$3$   & .466 & .494&.485\\ 
   HO& .5       &$4$  & .861 & .910 &.904\\    
  BL& .8       &$2$   & .129 & .153 &.150\\ 
   BL& .8     &$3$   & .472 & .503 &.499\\ 
   BL& .8      &$4$  & .842 & .863 &.862\\  
    NE& -.01  &$2$      & .120& .130 &.130 \\    
   NE& -.01   &$3$    & .598 & .617 &.617 \\  
   NE& -.01   &$4$     & .919 & .925 &.925 \\  \hline  
       \end{tabular}
\label{table:power3}
\end{center} 
\end{table}

Comparing Figure \ref{fig:powers} and Table \ref{table:power3} shows that for $\gamma=0.05$, the power of our method was roughly equal to that of BH, yet often slightly lower. However, our method provides simultaneous bounds and $\gamma$ can be chosen after seeing the results. As expected, the adaptive BH methods had a bit more power than BH. The adaptive methods performed similarly to each other. We found that they provided valid FDR control in all the settings, except in the settings ``HO", where the FDR of BH* varied around 0.08 and the FDR of BH** varied around 0.07. 


\section{Discussion}
This paper provides an exploratory multiple testing approach, which is useful in particular because the user is allowed to freely  use the data to choose rejection thresholds. 
This is what many researchers would like to do, but is not allowed by most popular methods.
We have provided a result on admissibility of our approach and simulations show good power,  especially in settings with many false hypotheses.  Moreover, the power properties can be influenced by the user, who may select an appropriate family of candidate envelopes $\mathbb{B}$. The choice of the range $\mathbb{T}$ of rejection thresholds also affects power, since the method focuses the power on the thresholds within this range.

Since our method essentially provides estimates for the FDP without confidence intervals, we encourage users to also compute a confidence interval, using e.g. the methods listed in the Introduction. However, as discussed, the methods among those that are valid under dependence  have limited power. This means that the confidence interval for the FDP may contain 1, even when there are several strong signals.
If permuting data is valid, this can often be used to construct tighter confidence intervals    \citep{hemerik2019permutation,blain2022notip,andreella2023permutation}.

Our simulations illustrate that for a given $\gamma$, BH tends to have slightly more power than our method, but our method has the advantage that it provides  post hoc inference. Indeed, we have shown that BH often becomes  too liberal when $\alpha$ is chosen post hoc. On the other hand, we control the median of the FDP, which may not always be as appealing as control of the mean. 
To further illustrate the utility of our method, in the Supplementary Material we provide a data analysis of real RNA-Seq data. Here we further explain that the  flexibility leads to added insights into the data.

Both our method and BH have certain proven, finite-sample, theoretical guarantees, in particular under independence. None of the methods are guaranteed to be valid under an unknown dependence structure. However, there is much evidence that BH is valid for many dependence structures. 
Likewise, we did not find a simulation setting where our method was invalid.

Besides FDP estimators, we have provided a novel $\pi_0$ estimator. We have discussed simulations where this estimator was used within an adaptive BH approach. Future work may more extensively assess our estimator in such settings. 
Further avenues for potential future research become apparent in the Supplementary Material. There we discuss more general estimates of $\pi_0$ and $V(t)$, which can be combined with the approach in Section \ref{secsimbounds} of constructing simultaneous mFDP bounds.

Note that ``uniform'' or `` simultaneous'' control usually means that the probability of a union of events is kept below some value \citep{genovese2004stochastic,meinshausen2006false,blanchard2020post,goeman2021only}. Since ``FDR control'' is not defined as controlling a probability, ``simultaneous FDR control'' is in that sense undefined. However, interestingly, Corollary 1 in \citet{katsevich2018controlling} provides what might be called  ``simultaneous FDR control'', assuming the \emph{p}-values are independent. In particular, there  $\alpha$ can be chosen post hoc, while still guaranteeing that the FDR is at most $\alpha$.

\setlength{\bibsep}{3pt plus 0.3ex}  
\def\bibfont{\small}  

\bibliographystyle{biblstyle}
\bibliography{bibliography}

\begin{thebibliography}{59}
\providecommand{\natexlab}[1]{#1}
\providecommand{\url}[1]{\texttt{#1}}
\expandafter\ifx\csname urlstyle\endcsname\relax
  \providecommand{\doi}[1]{doi: #1}\else
  \providecommand{\doi}{doi: \begingroup \urlstyle{rm}\Url}\fi

\bibitem[Andreella et~al.(2023)Andreella, Hemerik, Weeda, Finos, and
  Goeman]{andreella2023permutation}
Andreella, A., Hemerik, J., Weeda, W., Finos, L., and Goeman, J.
\newblock Permutation-based true discovery proportions for {fMRI} cluster
  analysis.
\newblock \emph{Statistics in Medicine. Online First version}, 2023.

\bibitem[Barber and Cand{\`e}s(2015)]{barber2015controlling}
Barber, R.~F. and Cand{\`e}s, E.~J.
\newblock Controlling the false discovery rate via knockoffs.
\newblock \emph{The Annals of Statistics}, 43\penalty0 (5):\penalty0
  2055--2085, 2015.

\bibitem[Basu et~al.(2021)Basu, Fu, Saretto, and Sun]{basu2021empirical}
Basu, P., Fu, L., Saretto, A., and Sun, W.
\newblock Empirical {bayes} control of the false discovery exceedance.
\newblock \emph{arXiv preprint arXiv:2111.03885}, 2021.

\bibitem[Benjamini and Hochberg(1995)]{benjamini1995controlling}
Benjamini, Y. and Hochberg, Y.
\newblock Controlling the false discovery rate: a practical and powerful
  approach to multiple testing.
\newblock \emph{Journal of the royal statistical society. Series B
  (Methodological)}, pages 289--300, 1995.

\bibitem[Benjamini and Yekutieli(2001)]{benjamini2001control}
Benjamini, Y. and Yekutieli, D.
\newblock The control of the false discovery rate in multiple testing under
  dependency.
\newblock \emph{Annals of statistics}, pages 1165--1188, 2001.

\bibitem[Bernhard et~al.(2004)Bernhard, Klein, and Hommel]{bernhard2004global}
Bernhard, G., Klein, M., and Hommel, G.
\newblock Global and multiple test procedures using ordered p-values—a
  review.
\newblock \emph{Statistical Papers}, 45\penalty0 (1):\penalty0 1--14, 2004.

\bibitem[Best et~al.(2015)Best, Sol, Kooi, Tannous, Westerman, Rustenburg,
  Schellen, Verschueren, Post, Koster, et~al.]{best2015rna}
Best, M.~G., Sol, N., Kooi, I., Tannous, J., Westerman, B.~A., Rustenburg, F.,
  Schellen, P., Verschueren, H., Post, E., Koster, J., et~al.
\newblock {RNA-Seq} of tumor-educated platelets enables blood-based pan-cancer,
  multiclass, and molecular pathway cancer diagnostics.
\newblock \emph{Cancer cell}, 28\penalty0 (5):\penalty0 666--676, 2015.

\bibitem[Blain et~al.(2022)Blain, Thirion, and Neuvial]{blain2022notip}
Blain, A., Thirion, B., and Neuvial, P.
\newblock Notip: Non-parametric true discovery proportion control for brain
  imaging.
\newblock \emph{NeuroImage}, 260\penalty0 (119492), 2022.

\bibitem[Blanchard et~al.(2020)Blanchard, Neuvial, Roquain,
  et~al.]{blanchard2020post}
Blanchard, G., Neuvial, P., Roquain, E., et~al.
\newblock Post hoc confidence bounds on false positives using reference
  families.
\newblock \emph{Annals of Statistics}, 48\penalty0 (3):\penalty0 1281--1303,
  2020.

\bibitem[Delattre and Roquain(2015)]{delattre2015new}
Delattre, S. and Roquain, E.
\newblock New procedures controlling the false discovery proportion via
  {Romano-Wolf}’s heuristic.
\newblock \emph{The Annals of Statistics}, 43\penalty0 (3):\penalty0
  1141--1177, 2015.

\bibitem[Dickhaus(2014)]{dickhaus2014simultaneous}
Dickhaus, T.
\newblock \emph{Simultaneous statistical inference: with applications in the
  life sciences}.
\newblock Springer Science \& Business Media, 2014.

\bibitem[Ditzhaus and Janssen(2019)]{ditzhaus2019variability}
Ditzhaus, M. and Janssen, A.
\newblock Variability and stability of the false discovery proportion.
\newblock \emph{Electronic Journal of Statistics}, 13\penalty0 (1):\penalty0
  882--910, 2019.

\bibitem[D{\"o}hler and Roquain(2020)]{dohler2020controlling}
D{\"o}hler, S. and Roquain, E.
\newblock Controlling the false discovery exceedance for heterogeneous tests.
\newblock \emph{Electronic Journal of Statistics}, 14\penalty0 (2):\penalty0
  4244--4272, 2020.

\bibitem[Efron(2007)]{efron2007correlation}
Efron, B.
\newblock Correlation and large-scale simultaneous significance testing.
\newblock \emph{Journal of the American Statistical Association}, 102\penalty0
  (477):\penalty0 93--103, 2007.

\bibitem[Farcomeni(2008)]{farcomeni2008review}
Farcomeni, A.
\newblock A review of modern multiple hypothesis testing, with particular
  attention to the false discovery proportion.
\newblock \emph{Statistical methods in medical research}, 17\penalty0
  (4):\penalty0 347--388, 2008.

\bibitem[Genovese and Wasserman(2004)]{genovese2004stochastic}
Genovese, C. and Wasserman, L.
\newblock A stochastic process approach to false discovery control.
\newblock \emph{Annals of Statistics}, pages 1035--1061, 2004.

\bibitem[Genovese and Wasserman(2006)]{genovese2006exceedance}
Genovese, C.~R. and Wasserman, L.
\newblock Exceedance control of the false discovery proportion.
\newblock \emph{Journal of the American Statistical Association}, 101\penalty0
  (476):\penalty0 1408--1417, 2006.

\bibitem[Goeman and Solari(2011)]{goeman2011multiple}
Goeman, J.~J. and Solari, A.
\newblock Multiple testing for exploratory research.
\newblock \emph{Statistical Science}, 26\penalty0 (4):\penalty0 584--597, 2011.

\bibitem[Goeman and Solari(2014)]{goeman2014multiple}
Goeman, J.~J. and Solari, A.
\newblock Multiple hypothesis testing in genomics.
\newblock \emph{Statistics in medicine}, 33\penalty0 (11):\penalty0 1946--1978,
  2014.

\bibitem[Goeman et~al.(2019)Goeman, Meijer, Krebs, and
  Solari]{goeman2019simultaneous}
Goeman, J.~J., Meijer, R.~J., Krebs, T.~J., and Solari, A.
\newblock Simultaneous control of all false discovery proportions in
  large-scale multiple hypothesis testing.
\newblock \emph{Biometrika}, 106\penalty0 (4):\penalty0 841--856, 2019.

\bibitem[Goeman et~al.(2021)Goeman, Hemerik, and Solari]{goeman2021only}
Goeman, J.~J., Hemerik, J., and Solari, A.
\newblock Only closed testing procedures are admissible for controlling false
  discovery proportions.
\newblock \emph{The Annals of Statistics}, 49\penalty0 (2):\penalty0
  1218--1238, 2021.

\bibitem[Gr{\"u}nwald(2022)]{grunwald2022beyond}
Gr{\"u}nwald, P.
\newblock Beyond {Neyman-Pearson}.
\newblock \emph{arXiv preprint arXiv:2205.00901}, 2022.

\bibitem[Guo and Romano(2007)]{guo2007generalized}
Guo, W. and Romano, J.
\newblock A generalized {Sidak-Holm} procedure and control of generalized error
  rates under independence.
\newblock \emph{Statistical applications in genetics and molecular biology},
  6\penalty0 (1), 2007.

\bibitem[Guo et~al.(2014)Guo, He, Sarkar, et~al.]{guo2014further}
Guo, W., He, L., Sarkar, S.~K., et~al.
\newblock Further results on controlling the false discovery proportion.
\newblock \emph{The Annals of Statistics}, 42\penalty0 (3):\penalty0
  1070--1101, 2014.

\bibitem[Harvey et~al.(2020)Harvey, Liu, and Saretto]{harvey2020evaluation}
Harvey, C.~R., Liu, Y., and Saretto, A.
\newblock An evaluation of alternative multiple testing methods for finance
  applications.
\newblock \emph{The Review of Asset Pricing Studies}, 10\penalty0 (2):\penalty0
  199--248, 2020.

\bibitem[Hemerik and Goeman(2018)]{hemerik2018false}
Hemerik, J. and Goeman, J.~J.
\newblock False discovery proportion estimation by permutations: confidence for
  significance analysis of microarrays.
\newblock \emph{Journal of the Royal Statistical Society: Series B (Statistical
  Methodology)}, 80\penalty0 (1):\penalty0 137--155, 2018.

\bibitem[Hemerik et~al.(2019)Hemerik, Solari, and
  Goeman]{hemerik2019permutation}
Hemerik, J., Solari, A., and Goeman, J.~J.
\newblock Permutation-based simultaneous confidence bounds for the false
  discovery proportion.
\newblock \emph{Biometrika}, 106\penalty0 (3):\penalty0 635--649, 2019.

\bibitem[Hoang and Dickhaus(2022)]{hoang2022usage}
Hoang, A.-T. and Dickhaus, T.
\newblock On the usage of randomized p-values in the {S}chweder--{S}pj{\o}tvoll
  estimator.
\newblock \emph{Annals of the Institute of Statistical Mathematics},
  74\penalty0 (2):\penalty0 289--319, 2022.

\bibitem[Hochberg and Benjamini(1990)]{hochberg1990more}
Hochberg, Y. and Benjamini, Y.
\newblock More powerful procedures for multiple significance testing.
\newblock \emph{Statistics in medicine}, 9\penalty0 (7):\penalty0 811--818,
  1990.

\bibitem[Hubbard(2004)]{hubbard2004alphabet}
Hubbard, R.
\newblock Alphabet soup: Blurring the distinctions between p's and alpha's in
  psychological research.
\newblock \emph{Theory \& Psychology}, 14\penalty0 (3):\penalty0 295--327,
  2004.

\bibitem[Katsevich and Ramdas(2018)]{katsevich2018controlling}
Katsevich, E. and Ramdas, A.
\newblock Towards ``simultaneous selective inference'': post hoc bounds on the
  false discovery proportion.
\newblock \emph{arXiv:1803.06790v3}, 2018.

\bibitem[Katsevich and Ramdas(2020)]{katsevich2020simultaneous}
Katsevich, E. and Ramdas, A.
\newblock Simultaneous high-probability bounds on the false discovery
  proportion in structured, regression and online settings.
\newblock \emph{The Annals of Statistics}, 48\penalty0 (6):\penalty0
  3465--3487, 2020.

\bibitem[Langaas et~al.(2005)Langaas, Lindqvist, and
  Ferkingstad]{langaas2005estimating}
Langaas, M., Lindqvist, B.~H., and Ferkingstad, E.
\newblock Estimating the proportion of true null hypotheses, with application
  to {DNA} microarray data.
\newblock \emph{Journal of the Royal Statistical Society: Series B (Statistical
  Methodology)}, 67\penalty0 (4):\penalty0 555--572, 2005.

\bibitem[Lehmann and Romano(2005)]{lehmann2005generalizations}
Lehmann, E.~L. and Romano, J.~P.
\newblock Generalizations of the familywise error rate.
\newblock \emph{The Annals of Statistics}, 33\penalty0 (3):\penalty0
  1138--1154, 2005.

\bibitem[Lei and Fithian(2018)]{lei2018adapt}
Lei, L. and Fithian, W.
\newblock {AdaPT}.
\newblock \emph{Journal of the Royal Statistical Society. Series B (Statistical
  Methodology)}, 80\penalty0 (4):\penalty0 649--679, 2018.

\bibitem[Lei et~al.(2021)Lei, Ramdas, and Fithian]{lei2021general}
Lei, L., Ramdas, A., and Fithian, W.
\newblock A general interactive framework for false discovery rate control
  under structural constraints.
\newblock \emph{Biometrika}, 108\penalty0 (2):\penalty0 253--267, 2021.

\bibitem[Li and Barber(2017)]{li2017accumulation}
Li, A. and Barber, R.~F.
\newblock Accumulation tests for {FDR} control in ordered hypothesis testing.
\newblock \emph{Journal of the American Statistical Association}, 112\penalty0
  (518):\penalty0 837--849, 2017.

\bibitem[Liang and Nettleton(2012)]{liang2012adaptive}
Liang, K. and Nettleton, D.
\newblock Adaptive and dynamic adaptive procedures for false discovery rate
  control and estimation.
\newblock \emph{Journal of the Royal Statistical Society Series B: Statistical
  Methodology}, 74\penalty0 (1):\penalty0 163--182, 2012.

\bibitem[Love et~al.(2014)Love, Huber, and Anders]{love2014moderated}
Love, M.~I., Huber, W., and Anders, S.
\newblock Moderated estimation of fold change and dispersion for {RNA-seq} data
  with {DESeq2}.
\newblock \emph{Genome biology}, 15\penalty0 (12):\penalty0 1--21, 2014.

\bibitem[Luo et~al.(2020)Luo, He, Emery, Noble, and Keich]{luo2020competition}
Luo, D., He, Y., Emery, K., Noble, W.~S., and Keich, U.
\newblock Competition-based control of the false discovery proportion.
\newblock \emph{arXiv preprint arXiv:2011.11939}, 2020.

\bibitem[Marcus et~al.(1976)Marcus, Eric, and Gabriel]{marcus1976closed}
Marcus, R., Eric, P., and Gabriel, K.~R.
\newblock On closed testing procedures with special reference to ordered
  analysis of variance.
\newblock \emph{Biometrika}, 63\penalty0 (3):\penalty0 655--660, 1976.

\bibitem[Meinshausen(2006)]{meinshausen2006false}
Meinshausen, N.
\newblock False discovery control for multiple tests of association under
  general dependence.
\newblock \emph{Scandinavian Journal of Statistics}, 33\penalty0 (2):\penalty0
  227--237, 2006.

\bibitem[Meinshausen et~al.(2006)Meinshausen, Rice,
  et~al.]{meinshausen2006estimating}
Meinshausen, N., Rice, J., et~al.
\newblock Estimating the proportion of false null hypotheses among a large
  number of independently tested hypotheses.
\newblock \emph{The Annals of Statistics}, 34\penalty0 (1):\penalty0 373--393,
  2006.

\bibitem[Miecznikowski and Wang(2022)]{mie2022}
Miecznikowski, J. and Wang, J.
\newblock Exceedance control of the false discovery proportion via high
  precision inversion method of {Berk} {Jones} statistics.
\newblock \emph{(Submitted)}, 2022.

\bibitem[Neuvial(2008)]{neuvial2008asymptotic}
Neuvial, P.
\newblock Asymptotic properties of false discovery rate controlling procedures
  under independence.
\newblock \emph{Electronic journal of statistics}, 2:\penalty0 1065--1110,
  2008.

\bibitem[{R Core Team}()]{R}
{R Core Team}.
\newblock \emph{R: A Language and Environment for Statistical Computing}.
\newblock R Foundation for Statistical Computing.
\newblock URL \url{https://www.R-project.org/}.

\bibitem[Rajchert and Keich(2022)]{rajchert2022controlling}
Rajchert, A. and Keich, U.
\newblock Controlling the false discovery rate via knockoffs: is the+ 1 needed?
\newblock \emph{arXiv preprint arXiv:2204.13248}, 2022.

\bibitem[Rogan and Gladen(1978)]{rogan1978estimating}
Rogan, W.~J. and Gladen, B.
\newblock Estimating prevalence from the results of a screening test.
\newblock \emph{American journal of epidemiology}, 107\penalty0 (1):\penalty0
  71--76, 1978.

\bibitem[Romano and Shaikh(2006)]{romano2006stepup}
Romano, J.~P. and Shaikh, A.~M.
\newblock Stepup procedures for control of generalizations of the familywise
  error rate.
\newblock \emph{The Annals of Statistics}, 34\penalty0 (4):\penalty0
  1850--1873, 2006.

\bibitem[Romano and Wolf(2007)]{romano2007control}
Romano, J.~P. and Wolf, M.
\newblock Control of generalized error rates in multiple testing.
\newblock \emph{The Annals of Statistics}, 35\penalty0 (4):\penalty0
  1378--1408, 2007.

\bibitem[Romano et~al.(2008)Romano, Shaikh, and Wolf]{romano2008formalized}
Romano, J.~P., Shaikh, A.~M., and Wolf, M.
\newblock Formalized data snooping based on generalized error rates.
\newblock \emph{Econometric Theory}, 24\penalty0 (2):\penalty0 404--447, 2008.

\bibitem[Roquain(2011)]{roquain2011type}
Roquain, E.
\newblock Type {I} error rate control for testing many hypotheses: a survey
  with proofs.
\newblock \emph{hal-00547965v2}, 2011.

\bibitem[Rosenblatt(2021)]{rosenblatt2021prevalence}
Rosenblatt, J.~D.
\newblock Prevalence estimation.
\newblock In \emph{Handbook of Multiple Comparisons}, pages 183--210. Chapman
  and Hall/CRC, 2021.

\bibitem[Schwartzman and Lin(2011)]{schwartzman2011effect}
Schwartzman, A. and Lin, X.
\newblock The effect of correlation in false discovery rate estimation.
\newblock \emph{Biometrika}, 98\penalty0 (1):\penalty0 199--214, 2011.

\bibitem[Schweder and Spj{\o}tvoll(1982)]{schweder1982plots}
Schweder, T. and Spj{\o}tvoll, E.
\newblock Plots of p-values to evaluate many tests simultaneously.
\newblock \emph{Biometrika}, 69\penalty0 (3):\penalty0 493--502, 1982.

\bibitem[Solari and Goeman(2017)]{solari2017minimally}
Solari, A. and Goeman, J.~J.
\newblock Minimally adaptive {BH}: A tiny but uniform improvement of the
  procedure of {Benjamini} and {Hochberg}.
\newblock \emph{Biometrical Journal}, 59\penalty0 (4):\penalty0 776--780, 2017.

\bibitem[Storey(2002)]{storey2002direct}
Storey, J.~D.
\newblock A direct approach to false discovery rates.
\newblock \emph{Journal of the Royal Statistical Society: Series B (Statistical
  Methodology)}, 64\penalty0 (3):\penalty0 479--498, 2002.

\bibitem[van~der Laan et~al.(2004)van~der Laan, Dudoit, and
  Pollard]{van2004augmentation}
van~der~ Laan, M.~J., Dudoit, S., and Pollard, K.~S.
\newblock Augmentation procedures for control of the generalized family-wise
  error rate and tail probabilities for the proportion of false positives.
\newblock \emph{Statistical applications in genetics and molecular biology},
  3\penalty0 (1):\penalty0 15, 2004.

\bibitem[Vesely et~al.(2023)Vesely, Finos, and Goeman]{vesely2023permutation}
Vesely, A., Finos, L., and Goeman, J.~J.
\newblock Permutation-based true discovery guarantee by sum tests.
\newblock \emph{Journal of the Royal Statistical Society. Series B (Statistical
  Methodology). Online First version}, 2023.

\end{thebibliography}

\appendix
\section{Overview of the supplementary material}
Section  \ref{secproofs} contains proofs of results in the main paper. Section \ref{secBHnotposthoc} provides simulation results that show that Benjamini-Hochberg does not have a flexible, post hoc interpretation.  Section \ref{secadsim} provides simulation results as in Section  \ref{secsims} of the main article, but with $m=10^4$ and $\pi_0=0.9$.
An analysis of real RNA-Seq data is in Section \ref{secdata}. 
Section \ref{seccompest} contains a theoretical comparison of the means of the estimators $\overline{\pi}_0'$ and $\hat{\pi}_0'$ from the main paper. 
Section \ref{secotherm} contains theory on additional flexible, median unbiased estimates of $\pi_0$ and FDP based on closed testing.

\section{Proofs of results} \label{secproofs}

\subsection{Proof of Theorem \ref{thmenvelope}}

\begin{proof}
If $r$ is a vector containing, say, $l_r$ \emph{p}-values, then we write $\R(r,t)=\{1\leq i \leq l_r: r_i<t\}$,  to make the dependence on both  the \emph{p}-values and the threshold explicit. Analogously we define ${V}(r,t)$,  $\overline{V}(r,t)$ and $\tilde{B}(r)$. We use the convention that  $\R(t)=\R(p,t)$, ${V}(t)={V}(p,t)$, $\overline{V}(t)=\overline{V}(p,t)$ and $\tilde{B}=\tilde{B}(p)$.

Let $E$ be the event $\{\tilde{B}(p)\geq \tilde{B}(1-q)\}$ and suppose $E$ holds.
Note that $$\overline{V}(1-q,t)= |\{1\leq i \leq N: 1-q_i>1-t\}| = |\{1\leq i \leq N: q_i<t\}|=R(q,t).$$
Thus
$$  \tilde{B}(1-q)= \min\Big\{B\in \mathbb{B}: \bigcap_{t\in \mathbb{T}} \big\{B(t)\geq R(q,t)  \big\} \Big\} .$$
Note that $V(q,t)=|\N\cap \R(t)|=R(q,t)$.
Hence, for every $t\in \mathbb{T}$, $$V(t)=R(q,t)\leq \tilde{B}(1-q,t)\leq \tilde{B}(p,t).$$
Since $\mathbb{P}(E)\geq 0.5$, it follows that 
$$\mathbb{P}\Big[\bigcap_{t\in \mathbb{T}} \big\{ V(t)\leq  \tilde{B}(p,t)\big\} \Big] \geq 0.5,$$
as was to be shown.

Now we show that $\tilde{B}'$ is also a confidence envelope. 
Assume $E$ holds. Then 
for every $l\in \mathbb{T}$, 
$[R(l)-\tilde{B}(l)]^+\leq S(l)$, where we recall that $S(l)=R(l)-V(l)$. Since $S$ is non-decreasing in $l$, it follows that $$\max\{[R(l)-\tilde{B}(l)]^+: l\in \mathbb{T}, l\leq t)\}\leq S(t)$$ for every $t\in \mathbb{T}$. Consequently $V(t)\leq \tilde{B}'(t)$ for every $t\in \mathbb{T}$. This is true whenever  $E$ holds. Since $\mathbb{P}(E)\geq 0.5$, it follows that $\tilde{B}'$ is a confidence envelope. It improves $\tilde{B}$ when $[R(\cdot)-\tilde{B}(\cdot)]^+$ in strictly decreasing somewhere on $\mathbb{T}$.

\end{proof}

\subsection{Proof of Theorem \ref{thmctmethod}}

\begin{proof}
We first define a closed testing procedure \citep{goeman2011multiple,goeman2021only}. For every $I\in\mathcal{M}$, 
write $\overline{V}_I'(t):=|\{i\in I: p_i\geq 1-t\}|$ and 
$$\tilde{B}_I:=\min\Big\{B\in \mathbb{B}: \bigcap_{t\in \mathbb{T}} \big\{B(t)\geq\overline{V}_I(t)\big\} \Big\}.$$ 
Define $\tilde{B}'_I$ analogously to $\tilde{B}'$ from Theorem \ref{thmenvelope}:
$$\tilde{B}_I'(t) := R_I(t) -\max\{[R_I(l)-\tilde{B}_I(l)]^+: l\in \mathbb{T}, l\leq t)\} ,$$
For each $I\in\mathcal{M}$, consider the following local test \citep{goeman2011multiple} for the intersection hypothesis $H_I=\cap_{i\in I}H_i$:  
$$\phi^{loc}_I=\mathbbm{1}\big(\exists t\in\mathbb{T}:\text{ } R_I(t)>\tilde{B}'_I(t)\big),$$
where $\mathbbm{1}(\cdot)$ denotes the indicator function.
The closed testing procedure \citep{goeman2011multiple} corresponding to these local tests is defined by the tests
$$\phi^{ct}_I=\min\{\phi^{loc}_J: I\subseteq J\in \mathcal{M}\}.$$ Note that this equals $\phi^{loc}_{I\cup\{1\leq i \leq m: p_i>1/2\}}$, where we used that $\mathbb{T}\subseteq[0,1/2)$. 
By definition, this equals
\begin{equation} \label{eqltest}
\mathbbm{1}\big(\exists t\in\mathbb{T}:\text{ } R_{I\cup\{1\leq i \leq m: p_i>1/2\}}(t)>\tilde{B}'_{I\cup\{1\leq i \leq m: p_i>1/2\}}(t)\big).
\end{equation}
Since $\mathbb{T}\subseteq[0,1/2)$, we have
$$R_{I\cup\{1\leq i \leq m: p_i>1/2\}}(t) = R_{I}(t), $$
$$\tilde{B}'_{I\cup\{1\leq i \leq m: p_i>1/2\}}(t)=\tilde{B}'_{I}(t).$$
Thus, \eqref{eqltest}, equals
$$\mathbbm{1}\big(\exists t\in\mathbb{T}:\text{ } R_I(t)>\tilde{B}'(t)\big).$$

By assumption, $\tilde{B}(q)=\tilde{B}_{\N}(p)$ is an envelope. Consequently, $\tilde{B}'_{\N}$ is an envelope, which follows from an argument analogous to the second part of the proof of Theorem \ref{thmenvelope}. Hence, the local test  $\phi^{loc}_{\N}$ rejects with probability at most $0.5$. 
By \citet[][p.588]{goeman2011multiple}, simultaneously over all $I\in\mathcal{M}$, $50\%$ confidence bounds  for $|\mathcal{N} \cap I|$ are
$$\max\big\{|A|: \emptyset\neq A \subseteq  I \text{ and } \phi^{ct}_A=0 \big\}=$$
$$\max\big\{|A|: \emptyset\neq A \subseteq I \text{ and } \mathbbm{1}[\exists t\in\mathbb{T}:\text{ } R_A(t)>\tilde{B}'(t)]=0 \big\},$$
which equals the quantity in expression \eqref{formctbound}.
\end{proof}

\subsection{Proof of Theorem \ref{thmequivct}}

\begin{proof}
We first prove the first claim of the theorem and then the second claim.
\\
\\
\noindent \emph{Proof of the first claim.}\\
Let $s\in\mathbb{T}$. We will show that  $\overline{B}(\R(s)) = \tilde{B}'(s)$.   If $R(s)=0$, then $\overline{B}(\R(s)) =0= \tilde{B}'(s)$. Now assume that $R(s)>0$. For $A\in\mathcal{M}$, we define $\phi^{ct}_A$ as in the proof of Theorem  \ref{thmctmethod}. We have $\overline{B}(\R(s))=$ 
 $$\max\big\{|A|: \emptyset \neq A\subseteq \R(s) \text{ and }\forall t\in\mathbb{T}: |\R_A(t)|\leq \tilde{B}'(t)\big\}=$$
 $$\max\big\{|A|: \emptyset \neq A\subseteq \R(s) \text{ and }\forall t\in\mathbb{T}: |\{i \in A: p_i\leq t\}| \leq \tilde{B}'(t)\big\} .$$
Let $r_1,...,r_{R(s)}$ be the \emph{p}-values with indices in $\R(s)$, sorted in descending order, i.e. $r_1\geq ...\geq r_{R(s)}.$
Note that the above equals
\begin{equation} \label{eq:s}
  \max\big\{1\leq j\leq R(s): \forall t\in\mathbb{T}: |\{1\leq i \leq j: r_i \leq t\}| \leq \tilde{B}'(t)\big\}.
  \end{equation}

  Note that for every  $1\leq j\leq R(s)$, on $\mathbb{T}\cap [r_j,s]$ the discrete  function $t\mapsto \tilde{B}'(t)$ does not increase faster than the discrete function $t\mapsto |\{1\leq i \leq j: r_i \leq t\}|$ does, by construction of  $\tilde{B}'(t).$ 
  Moreover, on $\mathbb{T}\cap [r_j,s]$, the function 
 $t\mapsto \tilde{B}'(t)$ does not have jumps at points other than $r_1$,..,$r_j$.
  Consequently, for $1\leq j\leq R(s)$,  the event $$\big \{\forall t\in\mathbb{T}: |\{1\leq i \leq j: r_i \leq t\}| \leq \tilde{B}'(t)\big \}$$ happens if and only if
  $$\big \{ |\{1\leq i \leq j: r_i \leq s\}| \leq \tilde{B}'(s)\big \},$$
 happens, i.e., if and only if   $\{j \leq \tilde{B}'(s)\big \}$ happens.
  Thus, \eqref{eq:s} equals
  $\max\{1\leq j\leq R(s):  j \leq \tilde{B}'(s)\}  = R(s)\wedge \tilde{B}'(s)=\tilde{B}'(s)$, as was to be shown.
\\
\\
\noindent \emph{Proof of the second claim.}\\
Assume the procedure from Theorem \ref{thmctmethod} is admissible, i.e., there do not exist bounds $\overline{B}^*(I)$ such that 
$$\mathbb{P}\Big[ \bigcap_{I\in\mathcal{M}} \big\{ |\N\cap I|\leq    \overline{B}^*(I) \big\}   \Big]\geq0.5$$
and such that $\overline{B}^*(I)\leq \overline{B}^*(I)$ for all $I\in\mathcal{M}$ and 
$\mathbb{P}\{\exists I\in\mathcal{M}:\overline{B}^*(I) < \overline{B}(I)\}>0.$
Suppose that $\tilde{B}'$ is not admissible. We will show that this leads to a contradiction, which finishes the proof. 

If $\tilde{B}'$ is not admissible, then there is an envelope $B:\mathbb{T}\rightarrow \mathbb{N}$ such that $B(t)\leq \tilde{B}'(t)$ for all $t\in\mathbb{T}$ and such that $\mathbb{P}\{\exists t\in\mathbb{T}:B(t)< \tilde{B}'(t)\}>0.$ For every $I\in\mathcal{M}$, define  $\overline{B}^*(I)$ as in expression  \eqref{formctbound}, but replacing $\tilde{B}'(t)$ by $B(t)$. Note that for all $I\in\mathcal{M}$, we have  $\overline{B}^*(I)\leq \overline{B}(I)$.

By the first claim, proven above, we know that for every $t\in \mathbb{T}$,
$\overline{B}^*(\R(t))=B(t)$. Since $B$ is a uniform improvement of $\tilde{B}'$, the envelope
$t\mapsto  \overline{B}^*(\R(t))$ is a uniform improvement of the envelope $t\mapsto  \overline{B}(\R(t))$. This means that with strictly positive probability, there exists an $I\in\mathcal{M}$ for which  $\overline{B}^*(I)< \overline{B}(I)$.
This means that the procedure  of Theorem \ref{thmctmethod} is not admissible, which contradicts our premise. Thus, $\tilde{B}'$ is admissible.
\end{proof}

\subsection{Proof of Proposition \ref{propcomputekappa}}

\begin{proof}
We will first show that   formula \eqref{reformulation_kappamin} holds if we define
$$ \kappa_0 = \max\big\{ \kappa\in(0,\infty]:  B^{\kappa}(s_1)\geq \overline{V}(s_1)      \big\},$$
$$ \kappa_i = \max\big\{ \kappa\in(0,\infty]:  B^{\kappa}(1-p_i)\geq \overline{V}(1-p_i)      \big\},$$
$1\leq i \leq m$.
We then show that this $\kappa_i$ is actually equal to $\frac{1-p_i+c}{\overline{V}(1-p_i)}$ (and analogously for $\kappa_0$), which finishes the proof.

\emph{First step.} By definition,
\begin{equation} \label{eq:herhkappa}
\kappa_{\text{max}}=\max\Big\{\kappa\in (0,\infty]:\text{ } \bigcap_{t\in \mathbb{T}} \big\{ B^{\kappa}(t) \geq  \overline{V}(t) \big\} \Big\}.
\end{equation}
Let $$A:=\Big\{s_1\Big\}\cup\Big\{1-p_i: 1\leq i \leq m \text{ and }1-p_i\in\mathbb{T}\Big\}.$$
Note that on $\mathbb{T}$, the non-decreasing, discrete, right-continuous function $t\mapsto  \overline{V}(t)$ has a jump at every $t\in \mathbb{T}$ for which there is a $1\leq j\leq m$ such that $t=1-p_j$. 

Consequently, the expression \eqref{eq:herhkappa} equals 
$$\kappa_{\text{max}}=\max\Big\{\kappa\in (0,\infty]:\text{ } \text{for all } s\in A:  \text{ }  B^{\kappa}(s) \geq  \overline{V}(s)  \Big\}=$$
$$\kappa_0\wedge\min\Big\{\kappa_i: 1\leq i \leq m \text{ and } 1-p_i\in \mathbb{T} \Big\},$$
where
$$\kappa_0=\max\big\{ \kappa\in(0,\infty]:  B^{\kappa}(s_1)\geq \overline{V}(s_1 )     \big\}$$
and for $1\leq i \leq m$
$$\kappa_i= \max\big\{ \kappa\in(0,\infty]:  B^{\kappa}(1-p_i)\geq \overline{V}(1-p_i)      \big\}. $$

\emph{Second step.}
We have
$$\kappa_i= \max\big\{ \kappa\in(0,\infty]:  |\{1\leq i \leq m: i\kappa - c \leq 1-p_i\}|\geq \overline{V}(1-p_i)      \big\} =$$
$$ \max\big\{ \kappa\in(0,\infty]:    [\overline{V}(1-p_i)]\cdot  \kappa - c \leq 1-p_i    \big\} = \frac{1-p_i+c}{\overline{V}(1-p_i)}$$ and analogously for $\kappa_0$.
\end{proof}

\subsection{Proof of Theorem \ref{fdpcontrol}}

\begin{proof}
Let $E$  be the event that for all $t\in \mathbb{T}$ we have $FDP(t)\leq B(t)/R(t)$.
Suppose $E$ holds. Let $\gamma\in[0,1]$. We will show that $FDP_{\gamma}\leq \gamma$. This will then finish the proof of the last statement of the theorem, since $\mathbb{P}(E)\geq 0.5$.
Note that $FDP_{\gamma}=FDP[t_{\text{max}}(\gamma)]\leq
 B[t_{\text{max}}(\gamma)]/R[t_{\text{max}}(\gamma)]
 \leq \gamma$, so that we are done.
\end{proof}

\subsection{Proof of Proposition \ref{propadaptedp}}

\begin{proof}
First of all, note that the minimum in \eqref{defpad} exists, since $B$ takes values in $\mathbb{N}$ by definition.

Let $E$  be the event that for all $t\in \mathbb{T}$ we have $FDP(t)\leq B(t)/R(t)$.
Let $r$ be the number of hypotheses that this procedure rejects, i.e., the number of $H_i$ with $ p_i^{\text{ad}}\leq \gamma $.
Note that there is a $t'\in \mathbb{T}$ such that $|\{1\leq i \leq m: p_i\leq t'\}|=r$ and such that $B(t')/R(t')\leq\gamma$. Hence $FDP(t')\leq\gamma.$ This holds simultaneously over all $\gamma\in[0,1]$, since $E$ holds. Since $\mathbb{P}(E)\geq 0.5$, this finishes the proof. 
\end{proof}

\subsection{Proof of Proposition \ref{rewriteadaptedp}}

\begin{proof}
Note that $[\max\{s_1,p_i\},s_2]$ is simply  the set  $\mathbb{T}\cap{[p_i,1]}$ from definition \eqref{defpad}.
The discrete function $t\mapsto B(t)/R(t)$ only has jumps downwards at points $t\in \{p_1,...,p_m\}$. Thus, on $\mathbb{T}=[s_1,s_2]$,  this function takes its minimum at some $t\in \{s_1,p_1,p_2,...,p_m\}$. Hence, to compute the minimum, it suffices to take the minimum over all  $t\in \big[\max\{s_1,p_i\},s_2\big] \cap\big\{s_1,p_1,p_2,...,p_m\big\}$, as was to be shown.
\end{proof}

\section{Benjamini-Hochberg with post hoc $\alpha$: a simulation study} \label{secBHnotposthoc}

\subsection{Simulation setup}
In this simulation study, we evaluated the expected value of $\mathbb{E}(FDP/\alpha)$ for BH, in a setting where $\alpha$ was chosen after seeing the data. 
The idea was to choose $\alpha$ in such a way as to mimic how a real human might choose $\alpha$ based on the data.
The results below indicate that $\mathbb{E}(FDP/\alpha)$ can be much larger than $1$ when $\alpha$ is chosen post hoc. This implies that conditional on $\alpha$, $\mathbb{E}(FDP/\alpha)$ is not generally bounded by 1, i.e., the conditional expectation of the FDP is not generally below $\alpha$. 

The distribution of the simulated data was as in Section \ref{secsims} in the main paper. We focused on the situation where there was no dependence between the 1000 \emph{p}-values. The number of false hypotheses varied between 0 and 40. The effect size $\Delta$ varied between 2 and 4.

The target $FDP$, $\alpha$, was chosen in a data-dependent way as follows. We fixed a non-negative integer $R_{\text{des}}$, which we call the \emph{desired number of rejections}. 
Let $p_{(1)}^{bh}\leq ...\leq p_{(1000)}^{bh}$ denote the FDR-adjusted \emph{p}-values computed by BH.
Thus $p_{(R_{\text{des}})}^{bh}$ was the $R_{\text{des}}$-th smallest  adjusted \emph{p}-value.

The value $\alpha$ was chosen to be the value in the interval $[0.01,0.2]$ that was closest to $p_{(R_{\text{des}})}^{bh}$. For example, if $p_{(R_{\text{des}})}^{bh}=0.06$, then $\alpha=0.06; $ if $p_{(R_{\text{des}})}^{bh}=0.007$, then $\alpha=0.01$; and if  $p_{(R_{\text{des}})}^{bh}=0.38$, then $\alpha=0.2$. Note that if $p_{(1)}^{bh}$ was larger than 0.2, then $p_{(1)}^{bh}>\alpha$, so that there were no rejections and the FDP was 0.

The reason behind choosing $\alpha$ in this way is that a human might choose $\alpha$ in a similar fashion. A person often has an a priori idea of what the desired number of rejections is and may want to choose $\alpha$ accordingly. Moreover, a person would perhaps not want to use a very large or very small value for $\alpha$. That is why we restricted $\alpha$ to be in the set  $[0.01,0.2]$.

\subsection{Simulation results}
The simulation results are shown in Table \ref{table:BHph}. Here, for different settings, the estimate of $\mathbb{E}(FDP/\alpha)$ is shown. Each estimate was based on $5\cdot 10^4$ repeated simulations.
The results shown are based on simulations with $R_{\text{des}}=20$. For other values of $R_{\text{des}}$, we obtained comparable results.

\begin{table}[h!]
\caption{ The expected value of $FDP/\alpha$, depending on the total number of false hypotheses  and the effect size  $\Delta$. 
} 
\begin{center}
    \begin{tabular}{ l l l l l }    
\hline \\[-0.4cm]
  &  \multicolumn{3}{l}{\qquad \quad  $\Delta$}\\ \cline{2-4} 
Nr. of false hyp. &   2 \quad & 3 \quad & 4 \\ \hline 
0  & 3.23 & 3.23 & 3.23  \\ 
5   & 2.80 & 2.06  & 1.79\\ 
20   & 2.00 & 1.21 & 1.30 \\ 
40 &  1.40 & 1.17 & 0.98 \\  \hline  
       \end{tabular}
\label{table:BHph}
\end{center}
\end{table}

Note that in most of the simulations settings, $\mathbb{E}(FDP/\alpha)$ was estimated to be larger than 1. There is a clear pattern in the results: if there were few false hypotheses and small effect sizes, then $\mathbb{E}(FDP/\alpha)$ was the largest. In particular, it can be seen that for the setting where all hypotheses were true, the estimate of $\mathbb{E}(FDP/\alpha)$ was 3.23. This means that in this setting, conditional on the  post hoc chosen $\alpha$, $\mathbb{E}(FDP)$ often greatly exceeds $\alpha$. 

We see that when there were many strong effects, then $\mathbb{E}(FDP/\alpha)$ was smaller. This is as expected, since in those settings, most of smallest FDR-adjusted  \emph{p}-values tend to be correspond to false hypotheses. Likewise, note that if all 1000 hypotheses would be false, then  $\mathbb{E}(FDP/\alpha)=0$ of course.

\section{Additional simulations} \label{secadsim}

The present section provides simulation results as in Section \ref{secsims} of the paper, but with $m=10^4$, i.e., with 10 times more hypotheses. Other than that the settings were the same and $\pi_0$ was still $0.9$, so that there were now $10^3$ false hypotheses. The results are in Figure \ref{fig:powersupp}. Each estimate in the figure is based on $10^3$ repeated simulations.
Note that in this setting, our novel method performed somewhat better than CT$+$Simes overall. 

The performance of K\&R was comparable to that of the novel method.  As illustrated in Figure \ref{fig:powerextra} of the main article, the novel method usually beats K\&R if $\pi_0<0.9$. The reason is that  K\&R is not adaptive, i.e., it is quite conservative when $\pi_0$ is far from 1.

\begin{figure}[ht] 
\centering
  \includegraphics[width=\linewidth]{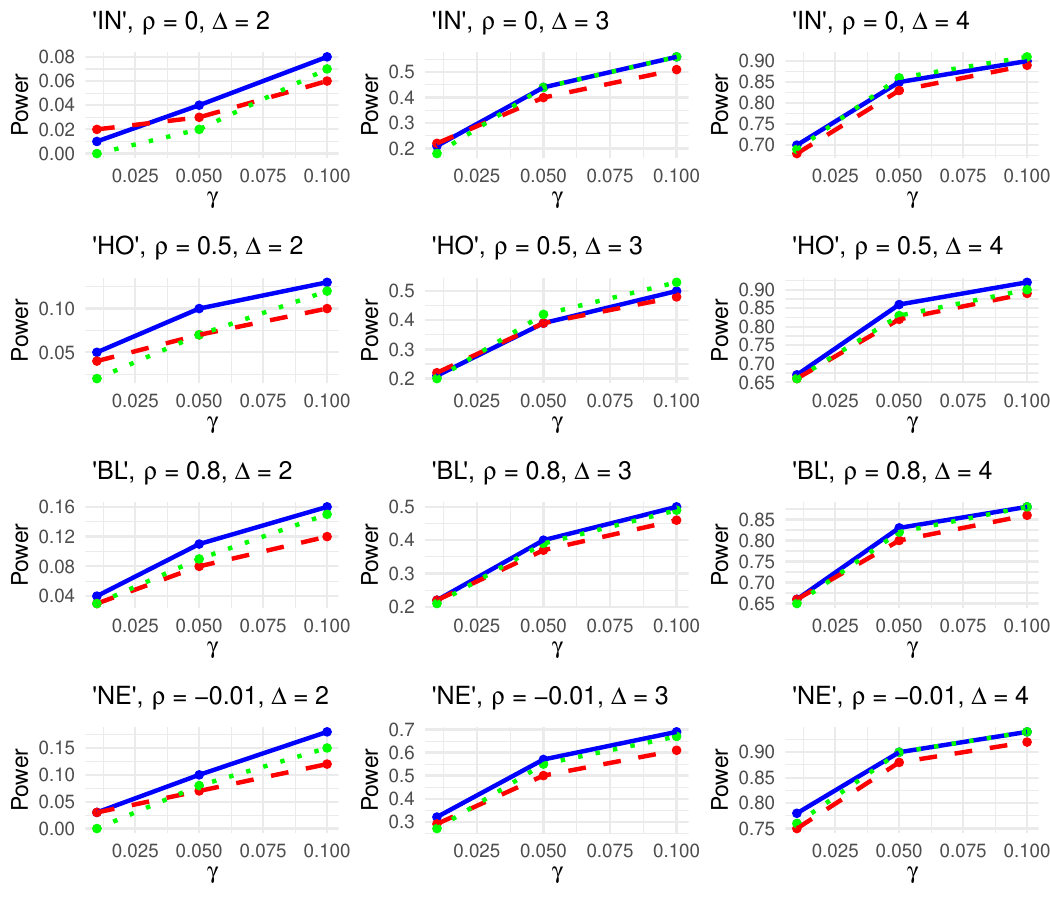}
  \caption{The power of our novel method (solid lines), CT$+$Simes (dashed lines) and K\&R (dotted lines)  as depending on $\gamma$, for various settings with $m=10^4$ and $\pi_0=0.9$.} \label{fig:powersupp}
\end{figure}

\section{Data analysis} \label{secdata}

We analyzed part of the RNA-Seq count data discussed in \citet{best2015rna}. The data are from 283 blood platelet samples. We downloaded the data from the Gene Expression Omnibus, accession GSE68086. 
The samples are from patients with one of six types of cancer, as well as controls. We used  the data from the 35 patients with pancreatic cancer and  the 42 patients with colorectal cancer (n=77).  The data contain read counts of 57736 transcripts. We removed the data on transcripts for which more than $75\%$ of the counts were 0, resulting in data on 10042 transcripts.

For each of the $m=10042$ transcripts, we tested the null hypothesis of no association with type of cancer, i.e, pancreatic or colorectal. To compute the $m$ raw \emph{p}-values, we used the R package DE-Seq2, which is currently the most standard approach \citep{love2014moderated}. This method employs a negative binomial model for each transcript.

We used the computed raw \emph{p}-values as input for our method of Sections \ref{secspecific}-\ref{secadjp}.
The method requires choosing the tuning parameter $c$ a priori.  
We chose $c=1/(2m)$, as recommended in Section \ref{secspecific} and used in the simulations. 
Figure \ref{fig:dataan} shows the number of rejections and the simultaneous bound for the number of false positives, as functions of the rejection threshold $t$. Note that for small thresholds, $R(t)$ is much larger than $\tilde{B}(t)$, so that the $mFDP$ is small. 
Note that roughly speaking, $R(t)$ increased faster in $t$ than $\tilde{B}(t)$. Consequently, the improved envelope $\tilde{B}'(t)$ from the second part of Theorem \ref{thmenvelope} was almost identical to  $\tilde{B}(t)$.
  In Figure \ref{fig:dataanFDP} the corresponding simultaneous $50\%$-confidence upper bounds for the FDP are shown.

\begin{figure}[ht] 
\centering
  \includegraphics[width=0.8\linewidth]{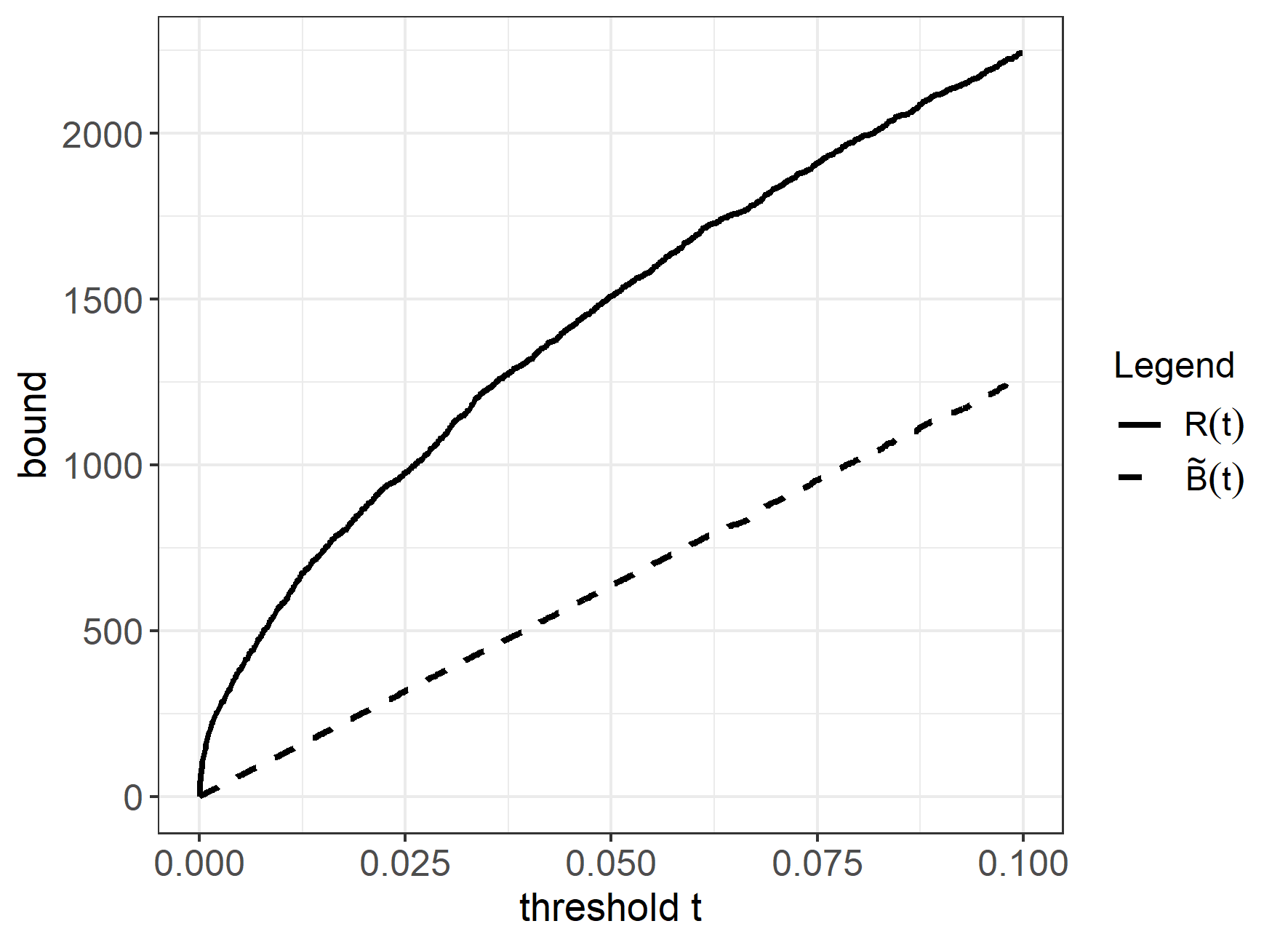}
  \caption{The number of rejections and the simultaneous bound for the number of false positives, as functions of the rejection threshold $t$.}
\end{figure} \label{fig:dataan}

\begin{figure}[ht] 
\centering
\captionsetup{width=0.8\linewidth}

\pgfplotsset{every axis/.append style={
                    label style={font=\Large},
                    tick label style={font=\Large}  
                    }}

\begin{tikzpicture}[scale=0.75]
\begin{axis}[
anchor=origin, 
width=12.5cm, height=10cm,
  axis y line = none,
  axis x line = top,
  xmin=0.00001, xmax=0.1, xlabel={Number of rejections $R(t)$}, ymin=0, ymax=0.6, xtick={0.001,0.01,0.03,0.05,0.07,0.09}, xticklabels={176, 582,1096,1510,1836,2120},
  font=\Large   
]
\end{axis}
\begin{axis}[
    anchor=origin,
                legend style={at={(12.5cm,.7)}, anchor=west, legend columns=1, draw=none},
    width=12.5cm, height=10cm,
axis x line=bottom, xlabel={Cut-off $t$}, xmin=0.00001, xmax=0.1,  xtick={0.001,0.01,0.03,0.05,0.07,0.09}, xticklabels={.001,0.01,0.03,0.05,0.07,0.09}, 
       axis y line=left, ylabel={Bound $\overline{FDP}(t)$}, ymin=0, ymax=0.6
       ]

\addplot[ line width = 1, very thick] table[x="cutoffs",y="mFDPc2", col sep=comma,smooth]{FDPdata.csv};
\end{axis}
\end{tikzpicture}
\caption{The  simultaneous $50\%$-confidence upper bound $\overline{FDP}(t):=\tilde{B}'(t)/R(t)$ as function of the rejection threshold $t$. For several values of $t$,  $R(t)$ is shown at the top of the graph. }
\label{fig:dataanFDP}
\end{figure}

We used Algorithm \ref{alg1} to compute mFDP-adjusted \emph{p}-values. 
These are useful, because the number of hypotheses that the method rejects can be computed as the number of adjusted \emph{p}-values that are at most $\gamma$.
We used the adjusted  \emph{p}-values for  generating Table \ref{table:dataan}, where the number of rejections is shown for various values of the mFDP-threshold $\gamma$. For comparison, we also show the number of rejections with BH for $\alpha=\gamma$. Note, however, that BH only allows using one value for $\alpha$, which moreover needs to chosen before seeing the data.

\begin{table}[!ht] \normalsize  
\caption{The number of rejected hypotheses for different values of $\gamma$, for two methods.
The first method is our procedure for simultaneous mFDP control. The inferences with this method are simultaneous, which means that with $50\%$ confidence, $FDP_{\gamma}\leq \gamma$ for all $\gamma\in\mathbb{T}$ simultaneously. The second method is BH, which ensures that if $\alpha$ is chosen prior to the data analysis, then $FDR=FDR(\alpha)\leq \alpha$, but not if $\alpha$ is chosen after inspecting the data.
} 
\begin{center}
    \begin{tabular}{ l l l l l}    
\hline \\[-0.4cm]
  &  \multicolumn{3}{l}{\qquad  $\gamma$ (i.e., $\alpha$)}\\ \cline{2-4} 
Method &    0.01 \quad & 0.05 \quad & 0.1 \quad \\ \hline 
mFDP control   & 24 & 125 & 243 \\ 
FDR control with BH  &  12 & 163 & 287 \\  \hline  
       \end{tabular}
\label{table:dataan}
\end{center}
\end{table}

The interpretation of the table is as follows.
If the user  first chooses e.g. $\gamma=0.05$, then she can reject 125 hypotheses. This means that with probability at least $50\%$, the true FDP is below 0.05 if we reject the 125 hypotheses with the smallest \emph{p}-values. Based on this promising result, the user may wonder how many hypotheses are rejected when $\gamma$ is decreased to 0.01. She finds that then 24 hypotheses are rejected. Since $\gamma=0.01$, the FDP is below 0.01 with probability at least $50\%$. The FDP must be a multiple of $1/24$,  so that it follows that with probability $50\%$, there are no false positives when these hypotheses are rejected. Thus, if, hypothetically, we repeat the experiment many times, then in at least $50\%$ of the cases, for all values of $\gamma$ that the user considers,  the $FDP=FDP_{\gamma}$ will be below  $\gamma$.

\section{Theoretical comparison of  $\overline{\pi}_0'$ and $\hat{\pi}_0'$}  \label{seccompest}

 The following result  says that often,  
  $\mathbb{E}(\overline{\pi}_0')\geq  \mathbb{E}(\hat{\pi}_0')$ if $t\in(0,0.5)$ and 
 $\mathbb{E}(\overline{\pi}_0')\leq  \mathbb{E}(\hat{\pi}_0')$ if $t\in(0.5,1)$. The difference between the expected values is often small, but usually strictly positive. Note that $t$ is often taken larger than $0.5$ in practice and then our estimator is less conservative than Schweder-Spj{\o}tvoll-Storey.
  
  \begin{proposition} \label{mean_of_estimator}
  Assume that all p-values have non-increasing densities or that  both $\mathbb{E}|\{i:p_i = 1-t\}|=0$ and
  \begin{equation} \label{decr}
  \frac{\mathbb{E}\big( |\{1\leq i \leq m: p_i> 1-t' \}|\big) }{t'} \leq
  \frac{\mathbb{E}\big( |\{1\leq i \leq m: p_i> t' \}|\big)}{1-t'},
  \end{equation}
  where $t'=\min\{t,1-t\}$
  Then, $\mathbb{E}(\overline{\pi}_0')\geq  \mathbb{E}(\hat{\pi}_0')$ if  $0\leq t <0.5 $   and $\mathbb{E}(\overline{\pi}_0')\leq  \mathbb{E}(\hat{\pi}_0')$ if $ 0.5<t\leq 1$. These inequalities are strict if the inequality \eqref{decr} is strict.

  \end{proposition}

 \begin{proof}

We have  
$$\hat{\pi}_0'-\overline{\pi}_0'= \frac{|\{i:p_i> 1-t\}|}{mt} -  \frac{ |\{i:p_i \geq 1-t\}| + |\{i:p_i> t\}|  }{m} =$$
$$  \frac{|\{i:p_i>1-t\}|}{mt} -  \frac{t\big(|\{i:p_i \geq 1-t\}| + |\{i:p_i>t\}|\big )}{mt} = $$
\begin{equation} \label{difestimates}
  \frac{ (1-t) |\{i:p_i > 1-t\}|  -t |\{i:p_i =1-t\}|   - t |\{i: p_i> t\}|   }{mt}.
\end{equation}

For $t=1/2$, this is 0. Now suppose $t \in (0,1/2)$.
If the densities $f_i$ of the \emph{p}-values $p_i$ are   non-increasing, then
$$ \mathbb{E} (|\{i:p_i > 1-t\}|/t) =$$
$$t^{-1} \sum_{i=1}^m  \int_{1-t}^{1} f_i(x) dx \leq$$
$$ (1-t)^{-1} \sum_{i=1}^m  \int_{t}^{1} f_i(x) dx = $$
$$ \mathbb{E}( |\{i: p_i> t\}| )/(1-t).$$ 
Here, the inequality is due to fact that the average of $f_i(x)$ on $[1-t,1]$ is smaller than or equal to the average of  $f_i(x)$ on $[t,1]$, since $t<1-t$ and the $f_i$ are non-increasing.
Multiplying both sides by $t(1-t)$ gives
$$\mathbb{E} \Big(\lambda |\{i:p_i > 1-t\}|    - t |\{i: p_i > t\}|\Big)\leq 0,$$
so that the expected value of \eqref{difestimates} is at most $0$.
In case $t\in(1/2,1)$, an analogous proof shows that the expected value of \eqref{difestimates} is at least $0$.

 \end{proof}

\section{Other methods for estimation of $\pi_0$ and FDP} \label{secotherm}
The main paper contains a useful framework for estimation of $\pi_0$ and the FDP, inspired by the Schweder-Spj{\o}tvoll-Storey estimator of $\pi_0$. For every cut-off $t$ considered, the paper derives a median unbiased estimate $\overline{V}(t)$ of the number of false positives $V(t)$. In Theorem 3, these bounds were used to derive a \emph{confidence envelope}, which provides simultaneous $50\%$-confidence bounds for $FDP(t)$. This confidence envelope was in turn used to provide flexible mFDP control.

In this appendix, we employ existing results related to closed testing, to generalize the definition of $\overline{V}(t)$. We will obtain a wide range of novel median unbiased estimates of $\pi_0$ and $V(t)$. Using the approach developed in the main paper, the novel estimates $\overline{V}(t)$ could also be used to provide novel confidence envelopes and FDP controlling procedures. These would be a generalization of the  methods developed in the main paper.

Note that Theorem 4 in the main paper also relies on closed testing theory, but there we directly constructed simultaneous bounds. Below, however, we construct bounds for fixed $t$, which could subsequently be used  to construct envelopes.

The procedures that we will derive, vary in terms of properties such as accuracy and bias. These properties always depend on the distribution of the data and there is no method that is uniformly best. 
We consider the method developed in the main paper particularly valuable, because it is sensible and relatively simple. For this reason, we focus on that method in the main paper.


\subsection{The Schweder-Spj{\o}tvoll-Storey method and closed testing}  \label{secstoreyct}
In Section 2, we derived median unbiased estimators of $\pi_0$ and the FDP.
Here we first derive the same result, but from the perspective of \emph{closed testing} \citep{marcus1976closed,goeman2011multiple,goeman2021only}. This perspective will reveal the broad class of novel estimators.

We start by explaining what closed testing is and how it can be used to obtain median unbiased estimators. The closed testing principle goes back to \citet{marcus1976closed} and can be used to construct multiple testing procedures that control the family-wise error rate. \citet{goeman2011multiple} show that such  procedures  can be extended to provide confidence bounds for the number of true hypotheses in all sets of hypotheses simultaneously. They construct $(1-\alpha)100\%$- confidence upper bounds -- $(1-\alpha)$-bounds for short -- for the FDP, where $\alpha\in (0,1)$. In this paper, we always consider $\gamma=\alpha=0.5$.

Let $\C$ be the collection of all nonempty subsets of $\{1,...,m\}$.
For every  $I\in \C$ consider the intersection hypothesis
$H_I=\cap_{i\in I} H_i$. This is the hypothesis that all $H_i$ with $i\in I$ are true.
For every $I\in \C$, consider some \emph{local test} $\delta(I)$, which is $1$ if $H_I$ is rejected and $0$ otherwise. Assume the test $\delta(\N)$ has level at most $\alpha$, so that $\mathbb{P}(\delta(\N)\geq 1)$ is bounded by $\alpha$.
Define $$\X=\{I \in \C: \delta(J)=1 \text{ for all } I\subseteq  J\subseteq \C\}.$$
The general closed testing procedure rejects all intersection hypotheses $H_I$ with $I\in \X$. It is well-known that this procedure controls the familywise error rate \citep{marcus1976closed}.
In \citet{goeman2011multiple} it is shown that we can also use the set $\X$ to provide a  $(1-\alpha)$-confidence upper bound for the number of true hypotheses in any $I\in\C$.
They show that
$$t_{\alpha}(I):=\max\{J \subseteq I:   J\not\in \X\}.$$
is a  $(1-\alpha)$-confidence upper bound for $|\N\cap I|$. In fact, they show that the bounds $t_{\alpha}(I)$ are valid simultaneously over all $I\in\C:$
\begin{equation} \label{GSresult}
   \mathbb{P} \Bigg[\bigcap_{I\in\C}\Big\{     |\N\cap I| \leq t_{\alpha}(I)   \Big\} \Bigg]\geq 1-\alpha.
   \end{equation}
The proof is short: $H_{\N}$ is rejected with probability at most $\alpha$, and if it is not rejected, then $\X$ contains no (sets of indices of) true hypotheses, which implies that $ |\N\cap I| \leq t_{\alpha}(I) $ for all $I\in \C$.
A different method,  formulated in \citet{genovese2006exceedance}, turns out to lead to the same bounds. This was first noted in the supplementary material of \citet{hemerik2019permutation} and in \citet{goeman2021only}.

We now turn to a closed testing procedure inspired by the Schweder-Spj{\o}tvoll-Storey estimator, which will lead to the same estimates as obtained in the previous sections. We only assume that Assumption 1 is satisfied and, for convenience,  that $N>0$.
Let $\mathbbm{1}(\cdot)$ be the indicator function. For every $I \in \C$, consider the local test
$$\delta(I)= \mathbbm{1}\Bigg(\big|\big\{i \in I: p_i\leq t\big\}\big|> \big|\big\{i \in I: p_i\geq 1-t\big\}\big|  \Bigg)=
\mathbbm{1}\Bigg( W_I^-> W_I^+  \Bigg),
$$
where 
\begin{equation} \label{Wbasic}
W_I^-=\big|\big\{i \in I: p_i\leq t\big\}\big|, \quad W_I^+=\big|\big\{i \in I: p_i\geq 1-t\big\}\big|.
\end{equation}
Take $\alpha=0.5$. It follows from Assumption \ref{as1} that $\mathbb{P}(\delta(\N)=1)\leq \alpha$. Consequently the bounds $t_{\alpha}(I)$, $I\in \C$, are simultaneous $50\%$-confidence upper bounds, i.e., the inequality \eqref{GSresult} is satisfied for $\alpha=0.5$.
In particular, $t_{\alpha}(\{1,...,m\})$ is a bound for the total number of true hypotheses, $N$. 
For every $1\leq a \leq m$, $Q_a$ be the set of indices of the $a$ largest \emph{p}-values, with ties broken arbitrarily.
For $t\in (0.0.5]$ we have
$$t_{\alpha}(\{1,...,m\})=\max\{J \subseteq \{1,...,m\}:   J\not\in \X\}=
\max\{1\leq a \leq m: Q_a \not\in \X \}=$$
$$ \max\{1\leq a \leq m: W^-_{Q_a} \leq W^+_{Q_a} \} =$$ 
$$ \min\Big\{m, 2\cdot |\{1\leq i \leq m: p_i\geq 1-t \}|+ |\{1\leq i \leq m: t < p_i < 1-t \}| \Big\}=$$ 
$$ \min\Big\{m, |\{1\leq i \leq m: p_i > t \}|+ |\{1\leq i \leq m: p_i \geq 1-t \}| \Big\}.$$
By a similar argument, we get the same result when $t\in(0.5,0)$.
Dividing this estimate by $m$ gives precisely our estimate $\overline{\pi}_0'$.
Thus,  based on the closed testing principle we obtain the same bound as using the argument in Section 2.1.

Now let $t\in(0,1)$ be a threshold and consider the rejected set $\R(t)=\{1\leq i \leq m: p_i\leq t\}$. Then one can check that  $t_{\alpha}(\R)$ is precisely the bound $\overline{V}(t)$ from section \ref{secmedunb}. Thus, the closed testing principle gives the same estimate as obtained before. Below, we will consider alternative local tests $\delta$, to obtain different methods.

\subsection{Different median unbiased estimates}
We will now consider a more general class of local tests $\delta$, which lead to estimates different from the ones considered until now.
Consider any non-decreasing, data-independent function $\psi:[0,1/2]\rightarrow \mathbb{R}$. 
For every $I\in \C$, define
 \begin{equation} \label{Wgeneral}
 W^-_I=\sum_{I^-} \psi(|1/2-p_i|), \quad W^+_I=\sum_{I^+} \psi(|p_i-1/2|),
 \end{equation}
 where $I^-=\{i \in I: p_i\leq t\}$,  $I^+=\{i \in I: p_i\geq 1-t\}$. 
  This is a generalization of the definition of $W^-_I$ and $W^+_I$ from section \ref{secstoreyct}. Indeed, if we take $\psi\equiv 1$, then the definitions \eqref{Wbasic} and \eqref{Wgeneral} coincide.

We make the following assumption, which is a generalization of Assumption 1 in the main paper.

\begin{assumption} \label{asgeneral}
The following holds: 
\begin{equation} \label{eqasgeneral}
\mathbb{P}\Big\{ W^-_{\N} > W^+_{\N} \Big\}\leq 0.5.
\end{equation}
 (If $N=0$, assume nothing.)
\end{assumption}

In  case $\psi\equiv 1$, the above assumption is the same as Assumption 1 in the main paper. We noted in section 2.2 that  that assumption is satisfied in particular if $(q_1,...,q_N)$ and $(1-q_1,...,1-q_N)$ have the same distribution. Note that Assumption \ref{asgeneral} is then satisfied as well for general $\psi$. 

For every $I\in \C$ we now consider the general local test 
$$\delta(I)= \mathbbm{1}( W_I^-> W_I^+),$$ 
where $W_I^-$ and $W_I^+$ depend on $\psi$ as in the definition \eqref{Wgeneral}. This general local test defines a general closed testing method that depends on $\psi$. We again denote the collection of sets rejected by the closed procedure by $\X$.
Based on this general closed tesing procedure we obtain $\psi$-dependent bounds $t_{\alpha}(I)$. Like before we have
$$t_{\alpha}(\{1,...,m\})=
\max\{1\leq a \leq m: Q_a \not\in \X \}=\max\{1\leq a \leq m: W^-_{Q_a} \leq  W^+_{Q_a} \}= $$
\begin{equation} \label{generalformulata}
\max\{1\leq a \leq m: W^-_{Q_a} \leq  W^+_{\{1,...,m\}} \}.
\end{equation}
This is a general, $\psi$-dependent, median unbiased estimator of $N$. 


Now suppose we use a rejection threshold $t\in(0,1/2]$, i.e., we reject all hypotheses with indices in $\R(t)$. For every $1\leq a \leq R(t)$, define $Q_a^t$ to be the set containing the indices of the largest $a$ \emph{p}-values that are strictly smaller than $t$ (with ties broken arbitrarily).


\begin{proposition}
Under Assumption \ref{asgeneral}, for any $t\in [0,1]$, a median unbiased estimate of $V(t)$ is $$\overline{V}_{\psi}(t) :=\max\{1\leq a \leq R(t): W^-_{Q_a^t}\leq W^+_{\{i:p_i>1-t\}}\}.$$

Dividing this by $R(t)$ gives a median unbiased estimate for the FDP:
$$\mathbb{P}(FDP(t)\leq \overline{V}_{\psi}(t)/R(t))\geq 0.5.$$
\end{proposition}

\begin{proof}
A median unbiased estimate of $V(t)$ is 
\begin{equation} \label{rewriteta}
t_{\alpha}(\R(t))=\max\{|I|: I\subseteq \R(t) \text{ and }I\not\in \X \}=\max\{1\leq a \leq R(t): Q_a^t \not\in \X \}.\end{equation}
Note that $Q_a^t \not\in \X$ if and only if its superset $J:=Q_a^t \cup \{i: p_i\geq 1-t \}$ is not rejected by its local test $\delta(J)$, i.e. when $W^-_{J}\leq W^+_{J}$, i.e. when $W^-_{Q_a^t}\leq W^+_{\{i: p_i\geq 1-t\}}$.

Hence the quantity \eqref{rewriteta} is equal to
$$\max\{1\leq a \leq R(t): W^-_{Q_a^t}\leq W^+_{\{i: p_i\geq 1-t\}}\}.$$ 
\end{proof}


The bounds $\overline{V}_{\psi}(t)$ can be immediately used within the theorems in the main paper to obtain confidence envelopes and FDP controlling procedures. For good performance, it can be necessary to adapt the set $\mathbb{B}$ of candidate envelopes in an appropriate way depending on the choice of $\psi$.

We will now discuss  two new examples of functions $\psi$, namely $\psi(x)=x$ and $\psi(x)=x^2$.
If $\psi(x)=x$, then for $I\in \C$ we have
$$\delta(I)= \mathbbm{1}( W_I^-> W_I^+)=$$
$$\mathbbm{1}\Big[ \sum_{i \in I^-} 0.5-p_i >  \sum_{i \in I^+} p_i-0.5  \Big]= \mathbbm{1}\Big[ |I|^{-1}\sum_{i \in I} p_i <0.5 \Big].$$
Thus the local test  simply  checks whether the average of the \emph{p}-values with indices in $I$ is below $0.5$.

For $\psi(x)=x^2$, we local test is 
$$\delta(I) =\mathbbm{1}\Big[ \sum_{i \in I^-} (0.5-p_i)^2 >  \sum_{i \in I^+} (p_i-0.5)^2  \Big].$$

The function $\psi$ `weights' the \emph{p}-values, depending on how far they are from $1/2$. If, rather than $\psi\equiv 1$, we take $\psi(x)=x$ or $\psi(x)=x^2$, then the \emph{p}-values that are far from $1/2$ receive the most weight.  The choice of $\psi$ influences the bias and variance of the $\pi_0$ and $FDP$ estimates. We found using simulations (not shown) that using $\psi\equiv 1$ often leads to a smaller expected value of the estimator of $\pi_0$ than using $\psi(x)=x$ or $\psi(x)=x^2$, but often to higher variance. The former makes intuitive sense if one looks at the formula \eqref{generalformulata} of the estimator of $N$.

\end{document}